\documentclass[preprint,3p,12pt]{elsarticle}

\usepackage{amssymb}

\usepackage{amsfonts,amsmath,amssymb,amsthm}
\usepackage{bm,nicefrac}
\usepackage{tabularx}
\usepackage{listings}
\usepackage{enumerate}
\usepackage{float}
\usepackage{graphicx}
\usepackage[labelformat=simple]{subcaption}

\usepackage{paralist}
\usepackage[ruled,lined,algosection]{algorithm2e}
\usepackage[english]{babel}

\usepackage[defaultcolor=red,final]{changes}
\definechangesauthor[color=red]{R2-2}
\definechangesauthor[color=red]{R2-4}
\definechangesauthor[color=red]{R3-1}
\definechangesauthor[color=red]{R3-2}
\definechangesauthor[color=red]{R3-3}
\definechangesauthor[color=red]{A-1}
\definechangesauthor[color=red]{A-2}
\definechangesauthor[color=red]{A-3}

\usepackage{amsmath}
\usepackage{graphicx}
\usepackage[colorlinks=true, allcolors=blue]{hyperref}

\newtheorem{theorem}{Theorem}
\newtheorem{definition}{Definition}
\newtheorem{lemma}{Lemma}

\DeclareMathOperator*{\argmax}{arg\,max}

\def\clap#1{\hbox to 0pt{\hss#1\hss}}

\numberwithin{equation}{section}

\begin{document}

\begin{frontmatter}



\title{Relativistic Space-Charge Field Calculation by Interpolation-Based Treecode}


\author[desy,uhh]{Yi-Kai Kan\corref{author}}
\author[desy,uhh]{Franz X. Kärtner}
\author[tuhh]{Sabine Le Borne}
\author[tuhh]{Jens-Peter M. Zemke}

\address[desy]{Center for Free-Electron Laser Science CFEL, Deutsches Elektronen-Synchrotron DESY, Germany}
\address[uhh]{Department of Physics, University of Hamburg, Germany}
\address[tuhh]{Hamburg University of Technology, Institute of Mathematics, Germany}
\cortext[author] {Corresponding author.\\\textit{E-mail addresses:} yikai.kan@desy.de (Y.-K Kan), franz.kaertner@desy.de (F. X.~Kärtner), leborne@tuhh.de (S.~Le Borne), zemke@tuhh.de (J.-P. M.~Zemke)}

\begin{abstract}
Space-charge effects are of great importance in particle accelerator physics. In the computational modeling, tree-based methods are increasingly used because of their effectiveness in handling non-uniform particle distributions and/or complex geometries. However, they are often formulated using an electrostatic force which is only a good approximation for low energy particle beams. For high energy, \emph{i.e.}, relativistic particle beams, the relativistic interaction kernel may need to be considered and the conventional treecode fails in this scenario. In this work, we formulate a treecode based on Lagrangian interpolation for computing the relativistic space-charge field. Two approaches are introduced to control the interpolation error. In the first approach, a modified admissibility condition is proposed for which the treecode can be used directly in the lab-frame. The second approach is based on the transformation of the particle beam to the rest-frame where the conventional admissibility condition can be used. Numerical simulation results using both methods will be compared and discussed.
\end{abstract}


\end{frontmatter}


\section{Introduction}
Space-charge effects are very important in accelerator physics and can lead to many unwanted phenomena. For example, the space-charge force limits the intensity of the electron current emitted from the cathode inside the electron gun~\cite{luginsland2002beyond} and causes the broadening of ultrafast electron packets in the free-space propagation~\cite{siwick2002ultrafast,collin2005transverse}. In numerical simulations, grid-based methods have been the standard choice. Among all grid-based methods, the particle-in-cell (PIC) method is probably the most popular choice. PIC is a self-consistent model considering the field generation from the charged particle and the field-particle interaction~\cite{dawson1983particle}. In the framework of electromagnetic particle-in-cell (EM-PIC), the particle trajectory is used to obtain charge or current density over a spatial grid by a charge deposition scheme~\cite{esirkepov2001exact,umeda2003new}. With the charge and current density, the corresponding electromagnetic field is then evaluated by solving Maxwell's equations~\cite{yee1966numerical,liu1997pseudospectral}. In combination with a suitable numerical integrator~\cite{boris1970relativistic}, the solution of the particle field and a given external field are used to push the particles to their new states of motion. For non-relativistic particle beams, the electrostatic particle-in-cell (ES-PIC) method is usually used where Poisson's equation is solved to compute the particle field~\cite{birdsall2018plasma}. In addition to the electromagnetic model, the relativistic particle beam can be simulated by a quasi-static model~\cite{flottmann2003recent,qiang2017symplectic} which solves the electrostatic field in the rest-frame of the particle beam and applies the corresponding electromagnetic field in the lab-frame. However, PIC has several numerical issues despite of its popularity. For example, the standard EM-PIC based on the finite-difference time-domain (FDTD) method has numerical dispersion due to the approximation with finite-difference stencils~\cite{taflove2005computational}. The pseudo-spectral methods, \emph{e.g.} the pseudo-spectral time-domain (PSTD)~\cite{liu1997pseudospectral} method or the pseudo-spectral analytical time-domain (PSATD) method~\cite{haber1973advances} were used to mitigate this problem by evaluating the spatial derivative in the spectral domain. However, because of the usage of the fast Fourier transform (FFT), the pseudo-spectral method is computationally more demanding compared to FDTD and its performance cannot scale over many computing nodes~\cite{lehe2018review}. Besides, standard PIC uses a fixed size grid to discretize the spatial domain and is inefficient for non-uniform particle distributions~\cite{zhang2017fast}. 

The computation of the space-charge field can also be achieved by a direct $N$-body summation (also called brute-force method). One major advantage of this method is the consideration of coulomb collision effect, which is especially critical in some problems of the accelerator physics \cite{reiser2008theory,gordon2021point}. However, a direct $N$-body methods requires a computational cost of $\mathcal{O}(N^2)$ and may not be applicable if a large number of particles $N$ is considered. Therefore, many efforts have been devoted to the development of tree-based methods where particles are subdivided into a hierarchy of clusters and the hierarchical relation of each cluster is stored in a tree data structure called cluster tree. Treecodes~\cite{barnes1986hierarchical} rely on the approximation of particle-cluster interactions and have a complexity down to $\mathcal{O}(N\log N)$. The force field on each particle is computed through an independent tree traversal starting from the root cluster and the applicability of the particle-cluster approximation is determined by the multiple acceptance criterion (MAC) which is similar to the admissibility condition in the study of hierarchical (H-) matrices~\cite{boerm2010efficient,hackbusch2015hierarchical}. The fast multipole method (FMM)~\cite{greengard1987fast} can further reduce the complexity down to $\mathcal{O}(N)$ by considering the cluster-cluster interactions. In traditional FMMs, each cluster at the same level is covered by a cubic box of the same size and the cluster-cluster interaction list is determined by a cluster's neighbor boxes and its parent's neighbor boxes. Compared to treecodes, one major drawback of traditional FMMs is that the well-separation condition relies on bounding boxes of a fixed size. Such definition of the well-separation condition leads to two problems:
\begin{compactenum}
\item Unlike the MAC in treecode, the well-separation condition cannot be flexibly controlled.
\item It excludes the usage of tight bounding boxes (can be of rectangular shape) which have an adaptive size depending on the cluster and are favorable for non-uniform particle distribution.
\end{compactenum}
Therefore, traditional FMMs are inefficient to treat non-uniform particle distributions~\cite{capuzzo1998comparison}. There also exist hybrid FMMs merging the strengths from both treecode and traditional FMMs~\cite{warren1995portable,dehnen2002hierarchical,cheng1999fast}. One effort of hybrid FMMs is based on 
the dual tree traversal~\cite{warren1995portable,dehnen2002hierarchical} where the cluster-cluster interaction list is determined by traversing the source and target cluster trees simultaneously and the well-separation condition can be defined as flexible as the MAC.

There have been many efforts using tree-based methods to model Coulomb interaction in the study of non-relativistic charged particles, \emph{e.g.} plasma dynamics~\cite{pfalzner19943d,thomas2016treecode}, electron dynamics in ultrafast electron microscopy~\cite{zhang2015differential} and proton dynamics in synchrotrons~\cite{jones1998hybrid}. For the dynamics of relativistic electron beams, the relativistic interaction kernel needs to be considered as the particles' field lines get compressed in the transverse direction~\cite{jackson1999classical}. The evaluation of the relativistic kernel based on the brute-force method was used in some studies~\cite{sell2013attosecond,kochikov2019relativistic} and has been provided in some beam dynamics codes, \emph{e.g.}, the General Particle Tracer (GPT)~\cite{van1996general}. To the best of our knowledge, there are few studies on using tree-based methods to calculate the relativistic space-charge field; only some previous works proposed by Schmid \emph{et al.} can be found~\cite{schmid2018reptil,schmid2019simulating}. Their approach is based on the quasi-static model: the space-charge field is solved in the particle's inertial frame by using the FMM with a spherical harmonics expansion~\cite{kurzak2006fast}.  

In this study, we first introduce the general concept of a treecode. After that, based on the treecode proposed by Wang \emph{et al.}~\cite{wang2020kernel}, we formulate an interpolation-based treecode for computing the relativistic space-charge field. In particular, we propose two methods to control the interpolation error:
\begin{compactenum}
    \item Based on the analytic estimation of the interpolation error bound, a modified admissibility condition is derived so that the formulated treecode can be performed directly in the lab-frame.
    \item The system is first transformed to the rest-frame of the average particle momentum in which the particle field is computed by a treecode and is then transformed back to the lab-frame. By using the relativistic transformation to the rest-frame, the formulated treecode can work with the conventional admissibility condition.
\end{compactenum}
Our numerical results show that the treecode based on the modified admissibility condition has better accuracy than the treecode based on the relativity transformation when a particle beam with momentum spread is considered; an explanation is also provided. Besides, we demonstrate that the proposed treecode scales like $\mathcal{O}(N\log N)$.

\section{The idea of treecode}
In this section, we provide a 2D example to illustrate the idea of treecode. The example provided here can easily be generalized to 3D, all our computations are done in 3D.

In the $N$-body problem, the total force field $f$ experienced by a target point $\bm{x} \neq \bm{x}_j$ can be written as
\begin{equation}
    f(\bm{x}) = \sum_{j\in \widehat{S}} g(\bm{x}, \bm{x}_j)m_j
\label{eq:treecode:summation_force_field}
\end{equation}
where $j$ denotes the index of the source particle with the position $\bm{x}_j$ in a cluster $S$ and $\widehat{S}$ denotes the index set of the particles in $S$. Through an interaction kernel $g(\bm{x},\bm{x}_j)$, the source particle at $\bm{x}_j$ applies the force field with the magnitude proportional to a physical quantity $m_j$ of the particle.

When evaluating~\eqref{eq:treecode:summation_force_field}, we can divide the source particles contained in $S$ into several smaller clusters and split the summation corresponding to
\begin{equation}
    f(\bm{x}) = 
      \sum_{j\in \widehat{S}^{(1)}_{(1,1)}}g(\bm{x}, \bm{x}_j)m_j +
      \sum_{j\in \widehat{S}^{(1)}_{(2,1)}}g(\bm{x}, \bm{x}_j)m_j +
      \sum_{j\in \widehat{S}^{(1)}_{(1,2)}}g(\bm{x}, \bm{x}_j)m_j +
      \sum_{j\in \widehat{S}^{(1)}_{(2,2)}}g(\bm{x}, \bm{x}_j)m_j,
\end{equation}
\begin{figure}[th]
\centering
\begin{subfigure}[b]{0.4\textwidth}
    \includegraphics[width=\linewidth]{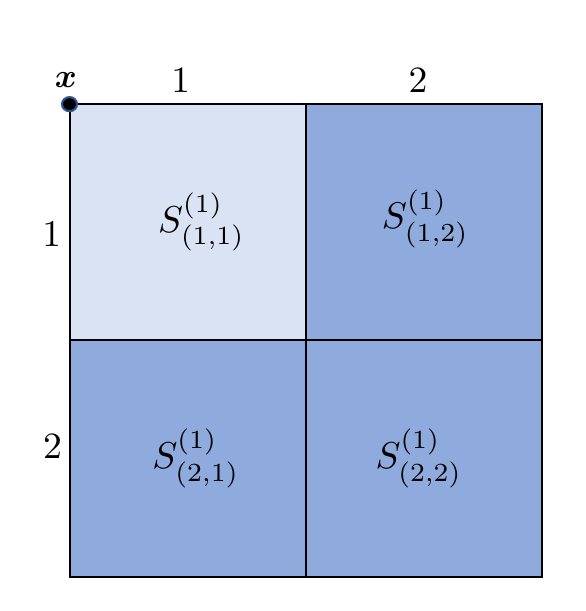}
    \caption{\label{fig:subdivision_illustration_a} $\ell=1$}
\end{subfigure}
\begin{subfigure}[b]{0.4\textwidth}
    \includegraphics[width=\linewidth]{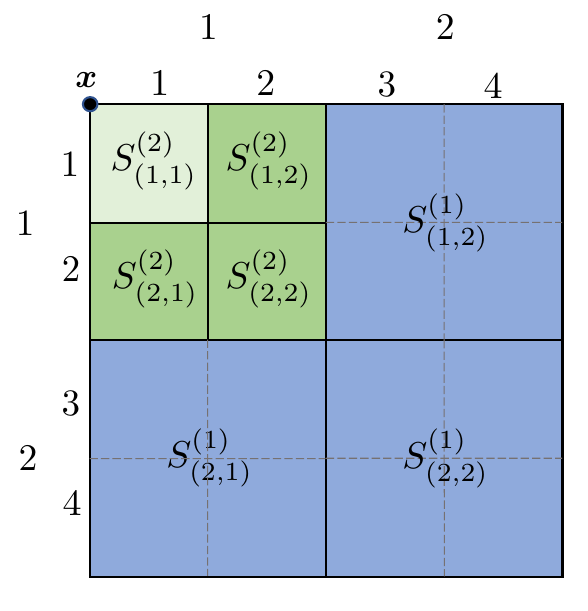}
    \caption{\label{fig:subdivision_illustration_b} $\ell=2$}
\end{subfigure}
\caption{Illustration of the subdivision of the particle cluster at the level (a) $\ell=1$ and (b) $\ell=2$. Here, the dark color denotes the far-field clusters and the light color denotes the near-field clusters.}
\label{fig:subdivision_illustration}
\end{figure}
where $S^{(1)}_{(1,1)}$, $S^{(1)}_{(2,1)}$, $S^{(1)}_{(1,2)}$ and $S^{(1)}_{(2,2)}$ are the sub-clusters located in the relative locations at the left-top, left-bottom, right-top, and right-bottom from the domain $Q_{S}$ of the cluster $S$, respectively (\hyperref[fig:subdivision_illustration_a]{Figure~\ref*{fig:subdivision_illustration_a}}). To express things in a general manner, we use $S^{(\ell)}_{(i,j)}$ to denote the sub-cluster from the $\ell$-th subdivision (also called level) of the root cluster $S$ with the location indices $i,j\in\{1,\ldots,2^{\ell}\}$.
If the distance from the target point $\bm{x}$ in $S^{(1)}_{(1,1)}$ to $S^{(1)}_{(2,1)}$, $S^{(1)}_{(1,2)}$, and $S^{(1)}_{(2,2)}$ is far enough (\emph{i.e.}, fulfills a sort of admissibility condition similar to the far-field condition) and the kernel function can be approximated by a separable expansion with a rank $N_r$ over the domain $Q_S$ of a cluster $S$ by
\begin{equation}\label{eq:separable_approximation}
  g(\bm{x}, \bm{x}_j)\approx \sum^{N_r}_{\nu=1}a_{S,\nu}(\bm{x})\cdot b_{S,\nu}(\bm{x}_j)\quad \bm{x}_j\in S,
\end{equation}
then the force field from the particles can be approximated by
\begin{align*}
    f(\bm{x})&=\sum_{j\in\widehat{S}^{(1)}_{(1,1)}}g(\bm{x},\bm{x}_j)m_j +
    \sum_{\tau\in T}\sum_{j\in\widehat{S}^{(1)}_{\tau}}g(\bm{x}, \bm{x}_j)m_j \\
    &\approx
    \sum_{j\in\widehat{S}^{(1)}_{(1,1)}}g(\bm{x},\bm{x}_j)m_j +
    \sum_{\tau\in T}\sum_{j\in\widehat{S}^{(1)}_{\tau}}\sum^{m}_{\nu=1}a_{S^{(1)}_{\tau},\nu}(\bm{x})\cdot b_{S^{(1)}_{\tau},\nu}(\bm{x}_j)m_j \\
    &=
    \sum_{j\in\widehat{S}^{(1)}_{(1,1)}}g(\bm{x},\bm{x}_j)m_j +
    \sum_{\tau\in T}\sum^{m}_{\nu=1}a_{S^{(1)}_{\tau},\nu}(\bm{x})\sum_{j\in\widehat{S}^{(1)}_{\tau}}b_{S^{(1)}_{\tau},\nu}(\bm{x}_j)m_j\\
    &=
    \sum_{j\in\widehat{S}^{(1)}_{(1,1)}}g(\bm{x},\bm{x}_j)m_j +
    \sum_{\tau\in T}\sum^{m}_{\nu=1}a_{S^{(1)}_{\tau},\nu}(\bm{x})\cdot m_{S^{(1)}_{\tau},\nu}
\end{align*}
with 
\begin{equation}\label{eq:macro_mass}
    m_{S^{(\ell)}_{\tau},\nu}:=\sum_{j\in\widehat{S}^{(\ell)}_{\tau}}b_{S^{(\ell)}_{\tau},\nu}(\bm{x}_j)m_j.
\end{equation}
Here, $T=\{(2,1),(1,2),(2,2)\}$ is a set of indices for the far-field clusters at each level. In this illustrative example, the target point $\bm{x}$ is in the top-left corner of the domain such that the indices of far-field clusters in each level all belong to the set $T=\{(2,1),(1,2),(2,2)\}$. In general, depending on the position of the target point, the index of a far-field cluster can be $(i,j)$ with $i,j\in\{1,\ldots,2^{\ell}\}$. The effective physical quantity of the macro particle with the index $\nu$ is denoted by $m_{S^{(\ell)}_{\tau},\nu}$ in~\eqref{eq:macro_mass}. The physical meaning of $m_{S^{(\ell)}_{\tau},\nu}$ is that the particles in a cluster are aggregated to a few macro particles; during the evaluation of the force field on the target particle, instead of traversing each real particle in the source cluster, we only consider a few macro particles if the approximation \eqref{eq:separable_approximation} is accurate, which is typically the case if the distance to the cluster is large enough.

To evaluate the force field from the source particles inside the cluster $S^{(1)}_{(1,1)}$, we can apply a similar trick as before. We first subdivide each sub-cluster from the level 1 into four sub-clusters ($S^{(2)}_{(1,1)}$, $S^{(2)}_{(2,1)}$, $S^{(2)}_{(1,2)}$, and $S^{(2)}_{(2,2)}$) and then use them to compute the approximated field (\hyperref[fig:subdivision_illustration_b]{Figure~\ref*{fig:subdivision_illustration_b}}). Assume that at each level $\ell$ of the subdivision only one cluster $S^{(\ell)}_{(1,1)}$ does not fulfill the far-field condition. We can compute the force field from $S^{(\ell)}_{(2,1)}$, $S^{(\ell)}_{(1,2)}$, and $S^{(\ell)}_{(2,2)}$ by the far-field approximation and subdivide $S^{(\ell)}_{(1,1)}$ to get the sub-clusters of the next level $\ell+1$. This procedure can be applied repeatedly until a maximum level $\kappa$ is reached and the approximated force field can be computed by
\begin{equation}
  \label{eq:intro_p2cluster_evaluation}
  f(\bm{x})\approx
  \underbrace{\sum_{j\in\widehat{S}^{(\kappa)}_{(1,1)}}g(\bm{x},\bm{x}_j)m_j}_{\text{near-field}} + 
  \underbrace{\sum^{\kappa}_{\ell=1}\sum_{\tau\in T}\sum^{N_r}_{\nu=1}a_{S^{(\ell)}_{\tau},\nu}(\bm{x})\cdot m_{S^{(\ell)}_{\tau},\nu}}_{\text{far-field}}.
\end{equation}
This means that the evaluation of the force-field can be decomposed into near-field and far-field terms. The far-field term is computed approximately by the effective physical quantities of the macro particles inside the far-field clusters of different levels. The near-field term is evaluated directly from the physical quantities of micro particles inside the near-field clusters. 

Consider a root cluster containing $N$ uniformly distributed particles.  
If every cluster of the finest level contains $N_0$ particles (\emph{i.e.}, $|\widehat{S}^{(\kappa)}_{(i,j)}|=N_0$ for all $i,j\in\{1,\ldots,2^{\kappa}\}$), the number of subdivisions is given by
\begin{equation*}
    \kappa=\dfrac{1}{2\log2}\log\left(\dfrac{N}{N_0}\right).
\end{equation*}
From a practical point of view it is reasonable to choose $N_0=N_r$ so that the number of macro particles $N_r$ in the clusters~\eqref{eq:separable_approximation} is not less than $N_0$. Since there are three far-field clusters at each level and one near-field cluster at the finest level, we need to traverse $1+3\kappa$ clusters to evaluate the field at $\bm{x}$. Therefore, the number of terms to evaluate~\eqref{eq:intro_p2cluster_evaluation} becomes
\begin{equation*}
    N_r\cdot(1+3\kappa) = N_r+\dfrac{3}{2\log 2}\cdot N_r\cdot\log\left(\dfrac{N}{N_r}\right).
\end{equation*}

If $N_r$ is bounded and independent of $N$, the computation cost for evaluating the force field experienced by a single particle becomes $\mathcal{O}\left(\log N\right)$. This argument can be applied to each particle in the root cluster $S$ and the total computation cost for evaluating the interaction force of $N$ particles in $S$ is $\mathcal{O}\left(N\log N\right)$.

From the example discussed above, we can see that the success of treecode relies on the hierarchical subdivision of particles into clusters and the approximation of the kernel function through a separable expansion. Thus, in the following sections, the following two main questions will be discussed:
\begin{compactitem}
    \item How can we subdivide the particles into a set of hierarchical clusters?
    \item How can we construct a separable approximation of a kernel function and bound the error of this approximation? 
\end{compactitem}

\section{Cluster tree}
A cluster tree~\cite{boerm2010efficient,hackbusch2015hierarchical} is a space-partition data structure (similar to the k-d tree~\cite{bentley1975multidimensional}) which provides an efficient way for finding the interaction list in the tree-based method. In this section, we briefly introduce the terminology and the construction of the cluster tree.

\begin{definition}[Cluster Tree]
A cluster tree is a labeled tree associated with an index set $I$ and fulfills the following requirements:
\begin{compactitem}
    \item For the root node $r$, its label set $\widehat{r}$ is given by the index set $I$, \emph{i.e.,} $\widehat{r}=I$.
    \item If $s$ is a non-leaf node, its label set is the union of the label sets of its children nodes (\emph{i.e.,} $\widehat{s}=\bigcup_{s'\in\mathrm{children}(s)}\widehat{s'}$).
    \item If $t',s'\in \text{children}(s)$, then $\widehat{s'}\bigcap\widehat{t'}=\emptyset$.
\end{compactitem}
We call the label sets associated with nodes of the tree (index)  \textit{clusters}.
\end{definition}

Let $S=\{\bm{x}_1,\ldots\bm{x}_j\}\subset\mathbb{R}^3$ be given by a set of particle positions (which we also call a \textit{cluster} of particles) with index set $\widehat{S}=\{1,\ldots,j\}$. We denote the entries in a vector by $\bm{x}=(x,y,z)^T$. The tightest rectangular box covering $S$ is called the \emph{bounding box} and is defined as
\begin{equation*}
    \text{bbox}(S):=[a_x,b_x]\times[a_y,b_y]\times[a_z,b_z],
\end{equation*}
with $a_{g}=\min_{i\in\widehat{S}}\{g_i\}$, $b_{g}=\max_{i\in\widehat{S}}\{g_i\}$ and $g\in\{x,y,z\}$. To subdivide a cluster $S$ containing $j$ particles, we first determine $\text{bbox}(S)$ which enables us to find the coordinate direction with the biggest interval $k=\argmax_{g\in\{x,y,z\}} (b_g - a_g)$. We then choose the particle position with the $\lfloor j/2\rfloor$\nobreakdash-th greatest value in the $k$-component of its position as the splitting point $\bm{x}_{\text{split}}$ such that the cluster $S$ can be subdivided into two sub-clusters
\begin{equation*}
  S_1=\{\bm{x}_i \mid k_i\leq k_{\text{split}},\ i\in\widehat{S}\}
  \quad\text{and}\quad 
  S_2=\{\bm{x}_i \mid k_i> k_{\text{split}},\ i\in\widehat{S}\},\quad
  S=S_1\bigcup S_2.
\end{equation*}
We can apply this subdivision procedure to each of the resulting sub-clusters repeatedly until the number of particles in the clusters at the deepest level is smaller than a pre-selected number $N_0$. Different to what has been used in the illustrative example of the previous section, this \replaced[id=A-2]{cardinality-balanced}{cardinality-based} subdivision strategy is used in the remainder of this study since it has the following benefits:
\begin{compactitem}
    \item Each subdivided cluster contains roughly the same number of particles so that the cluster tree is balanced regardless of the particle distribution.
    \item The subdivision based on the coordinate direction with the biggest interval can shrink the diameter of the cluster fast.
\end{compactitem}
 
\section{Approximation for the kernel function of the space-charge field}\label{kernel_function_apprximation}
The space-charge field from a relativistic particle at the position $\bm{x}_j$ experienced by a target position $\bm{x}$ can be expressed as~\cite{jackson1999classical,sell2013attosecond}
\begin{equation*}
    \bm{E}(\bm{x},\bm{x}_j, \bm{p}_j)=\dfrac{q}{4\pi\varepsilon_0}\gamma_j\bm{k}(\bm{x},\bm{x}_j,\bm{p}_j)
    \quad\text{and}\quad
    \bm{B}(\bm{x}, \bm{x}_j,\bm{p}_j)=\dfrac{q}{4\pi\epsilon_0 c_0} 
    \bm{p}_j\times \bm{k}(\bm{x}, \bm{x}_j,\bm{p}_j)
\end{equation*}
with
\begin{equation}\label{eq:treecode_exact_k_kernel}
    \bm{k}(\bm{x},\bm{x}_j,\bm{p}_j):=
    \dfrac{\bm{x}-\bm{x}_j}{\bigl| \Vert\bm{x}-\bm{x}_j\Vert^2_2 + \left(\bm{p}_j\cdot(\bm{x}-\bm{x}_j)\right)^2\bigr|^{3/2}}.
\end{equation}
Here, $\gamma_j=1/\sqrt{1-\Vert\bm{\beta}_j\Vert^2_2}$ and $\bm{p}_j=\gamma_j\bm{\beta}_j$ are the Lorentz factor and normalized momentum of the particle, respectively, with $\bm{\beta}_j$ the particle velocity normalized to the speed of light. Although not necessary, throughout this study, we assume that each particle in the system has the same charge $q$. In the physics of particle beams moving in the $z$-direction, the paraxial approximation can be usually applied
\begin{equation*}
    \overline{p}_z \gg \overline{p}_{\{x,y\}}
    \quad\text{and}\quad
    \overline{p}_{\{x,y,z\}} \gg \Delta p_{\{x,y,z\}}
\end{equation*}
where $\overline{\bm{p}} := \frac{1}{N}\sum^{N}_{j=1} \bm{p}_j$, $\bm{p}_{j} =: \overline{\bm{p}} + \Delta\bm{p}_{j}$ and $N$ is the number of particles.
Therefore, the kernel function can be approximated as
\begin{align}
    \bm{k}(\bm{x}, \bm{x}_j,\bm{p}_j)
&=\dfrac{\bm{x}-\bm{x}_j}{\bigl| \Vert\bm{x}-\bm{x}_j\Vert^2_2 + \left((\overline{\bm{p}}+\Delta\bm{p}_j)\cdot(\bm{x}-\bm{x}_j)\right)^2\bigr|^{3/2}} \nonumber\\
&\approx \dfrac{\bm{x}-\bm{x}_j}{\bigl| \Vert\bm{x}-\bm{x}_j\Vert^2_2 + \left(\overline{\bm{p}}\cdot(\bm{x}-\bm{x}_j)\right)^2\bigr|^{3/2}}.\label{eq:approx_k}
\end{align}
Since $\overline{p}_z\gg \overline{p}_{\{x,y\}}$, we may use $\overline{p}_z \approx \Vert\overline{\bm{p}}\Vert_2$ to further simplify equation \eqref{eq:approx_k} as in (cf.~\hyperref[appendx:approx_kernel_denominator]{Section~\ref*{appendx:approx_kernel_denominator}})
\begin{align}
    \bm{k}(\bm{x}, \bm{x}_j,\bm{p}_j)
    &\approx\dfrac{\bm{x}-\bm{x}_j}{\biggl((x-x_j)^2+(y-y_j)^2+(1+\overline{p}^2_z)(z-z_j)^2\biggr)^{3/2}}\nonumber\\
    &\approx\dfrac{\bm{x}-\bm{x}_j}{\biggl((x-x_j)^2+(y-y_j)^2+\overline{\gamma}^{2}(z-z_j)^2\biggr)^{3/2}}=:\bm{g}(\bm{x},\bm{x}_j) \label{eq:kernel_g}
\end{align}
where $\overline{\gamma}^2:=1+\Vert\overline{\bm{p}}\Vert^2_2$. The space-charge field from the source particle can be approximated by
\begin{equation*}
    \bm{E}(\bm{x}, \bm{x}_j) \approx 
    \dfrac{q}{4\pi\epsilon_0}
    \gamma_j\bm{g}(\bm{x}, \bm{x}_j)
    \quad\text{and}\quad 
    \bm{B}(\bm{x}, \bm{x}_j) \approx  
    \dfrac{q}{4\pi\epsilon_0 c_0}
    \bm{p}_j\times \bm{g}(\bm{x}, \bm{x}_j).
\end{equation*}

The separable approximation of the kernel function can be achieved by a tensor product interpolation with Lagrangian polynomials~\cite{boerm2010efficient,hackbusch2015hierarchical}
\begin{align*}
    \widetilde{\bm{g}}(\bm{x}, \bm{x}_j) 
    &= \sum^n_{\nu_1=0}\sum^n_{\nu_2=0}\sum^n_{\nu_3=0}\bm{g}\bigl(\bm{x},(\xi^{x}_{\nu_1}, \xi^{y}_{\nu_2}, \xi^{z}_{\nu_3})\bigr)\cdot\ell^{x}_{\nu_1}(x_j)\cdot\ell^{y}_{\nu_2}(y_j)\cdot\ell^{z}_{\nu_3}(z_j)\\
    &=:\sum_{\bm{\nu}\in\{0,\ldots,n\}^3}\bm{g}(\bm{x},\bm{\xi}_{\bm{\nu}})\ell_{\bm{\nu}}(\bm{x}_j),
\end{align*}
where the Lagrange basis polynomials are defined as
\begin{equation*}
    \ell_{j}^{k}(x) := \prod^{n}_{\substack{i=0\\i\neq j}}\dfrac{x - \xi_{i}^{k}}{\xi_{j}^{k} - \xi_{i}^{k}}, \quad j=1,\ldots,n, \quad k\in\{x,y,z\}
\end{equation*}
with the interpolation points $\xi_i^k$ and the polynomial degree $n$. 

The error bound of the tensor product interpolation of the kernel function with respect to the source point variable $\bm{x}_j$ in a rectangular domain $Q_{S}=[a_x,b_x]\times[a_y,b_y]\times[a_z,b_z]\subset \mathbb{R}^3$ is~\cite{hackbusch2015hierarchical}
\begin{equation}\label{eq:interpolation_error}
    \Vert \bm{g}(\bm{x},\cdot) - \widetilde{\bm{g}}(\bm{x},\cdot) \Vert_{\infty,Q_{S}}
    \leq 
    \text{const}\cdot\sum_{k\in\{x,y,z\}} (b_k - a_k)^{n+1}\cdot\dfrac{\Vert \partial^{n+1}_{k}\bm{g}(\bm{x},\cdot)\Vert_{\infty,Q_{S}}}{(n+1)!}.
\end{equation}
As indicated by~\eqref{eq:interpolation_error}, to control the interpolation error, we need to find a condition which guarantees a bound for the right-hand-side terms. In the literature, this condition is called admissibility condition~\cite{boerm2010efficient,hackbusch2015hierarchical}. In particular, when determining an admissibility condition for a kernel function $g(\bm{x},\cdot)$, we need to find a bound on the term  $\Vert \partial^{n+1}_{k}\bm{g}(\bm{x},\cdot)\Vert_{\infty,Q_{S}}$.

\section{Admissibility Condition for the Relativistic Kernel}
As mentioned in \hyperref[kernel_function_apprximation]{Section~\ref*{kernel_function_apprximation}}, to derive an admissibility condition for a kernel function, it is necessary to bound the higher-order derivatives of the kernel function. For studying the free space electrostatic problem, the force kernel can be written
\begin{equation}
    \bm{f}(\bm{x},\bm{x}_j):=
    \dfrac{\bm{x}-\bm{x}_j}{\Vert\bm{x}-\bm{x}_j\Vert^3_2}.
\end{equation}
An upper bound of the higher-order derivatives with respect to $\bm{x}_j$ is given by~\cite{hackbusch2015hierarchical}
\begin{equation*}
    \Vert\partial^{n+1}_{i}\bm{f}(\bm{x},\cdot)\Vert_{\infty,Q_{S}}\leq \dfrac{\text{const}}{\text{dist}(x,S)^2}\dfrac{(n+1)!}{\text{dist}(\bm{x},S)^{n+1}}\quad i\in\{x,y,z\}
\end{equation*}
and the substitution into ~\eqref{eq:interpolation_error} yields
\begin{equation}\label{eq:error_bound_es_kernel}
    \Vert \bm{f}(\bm{x}, \cdot) - \tilde{\bm{f}}(\bm{x}, \cdot) \Vert_{\infty,Q_{S}}
    \leq 
    \dfrac{\text{const}}{\text{dist}(\bm{x},S)^2}
    \left(\dfrac{\text{diam}(S)}{\text{dist}(\bm{x},S)}\right)^{n+1}.
\end{equation}
Here, we use the definitions
\begin{equation}\label{eq:definition_diam_dist}
    \text{diam}(S):=\max_{\bm{x}_i,\bm{x}_j\in S}\Vert \bm{x}_i-\bm{x}_j \Vert_2
    \quad\text{and}\quad
    \text{dist}(\bm{x},S):=\min_{\bm{x}_j\in S} \Vert \bm{x}-\bm{x}_j\Vert_2
\end{equation}
to denote the diameter of $S$ and the distance from $\bm{x}$ to $S$.

From \eqref{eq:error_bound_es_kernel}, we can define an admissibility condition (also called $\eta$-admissibility~\cite{hackbusch2015hierarchical}) for the electrostatic kernel by
\begin{equation*}
   \dfrac{\text{diam}(S)}{\text{dist}(\bm{x},S)} < \eta
\end{equation*}
with some admissibility parameter $\eta\in\mathbb{R}_{>0}$. 

However, the admissibility condition of the electrostatic kernel (called conventional admissibility condition in this study) may not be used for the relativistic kernel~\eqref{eq:kernel_g}. Thus, we will derive an upper bound of the interpolation error for the relativistic kernel and derive a corresponding admissibility condition. We first introduce a theorem which is useful for estimating the upper bound of the higher-order derivatives of a specific type of kernel function.
\begin{theorem}[\protect{\cite[Theorem~E.4]{hackbusch2015hierarchical}}]
\label{thm:norm_kernel_derivative}
Let $s(\bm{x}, \bm{x}_j)=1/\Vert\bm{x}-\bm{x}_j\Vert^{a}_{2}$ with $a\in\mathbb{R}_{>0}$, we have
\begin{equation*}
    |\partial^{\bm{\nu}}_{\bm{x}}s(\bm{x},\bm{x}_j)|\leq \bm{\nu}!w^{a/2+|\bm{\nu}|}\Vert\bm{x}-\bm{x}_j \Vert^{-|\bm{\nu}|-a}_2\quad \forall\bm{x},\bm{x}_j\in\mathbb{R}^d \land \bm{x}\neq\bm{x}_j,
\end{equation*}
with a suitable constant $w$. Here, $\bm{\nu}=(\nu_1,\ldots,\nu_d)\in\mathbb{N}^d$ and the following conventions are used \cite{hackbusch2015hierarchical}:
\begin{equation*}
    |\bm{\nu}|=\sum^{d}_{i=1}\nu_i,\quad 
    \bm{\nu}!=\prod^{d}_{i=1}\nu_{i}!,\quad \partial^{\bm{\nu}}_{\bm{x}}=\prod^{d}_{i=1}\Bigl(\dfrac{\partial}{\partial x_i}\Bigr)^{\nu_i}.
\end{equation*}
Since $s(\bm{x}, \bm{x}_j)$ is symmetric with respect to $\bm{x}$ and $\bm{x}_j$, the conclusion above also holds for the derivatives with respect to $\bm{x}_j$.
\end{theorem}

The force kernel for the relativistic space-charge field is \eqref{eq:kernel_g}
\begin{equation*}
    \bm{g}(\bm{x},\bm{x}_j)=\dfrac{\bm{x}-\bm{x}_j}{\biggl((x-x_j)^2+(y-y_j)^2+\overline{\gamma}^{2}(z-z_j)^2\biggr)^{3/2}}.
\end{equation*}
By introducing a change of variables 
\begin{equation*}
    \begin{cases}
        \overline{x}=x,\\
        \overline{y}=y,\\
        \overline{z}=\overline{\gamma}z,
    \end{cases}
    \quad
    \text{and}
    \quad
    \begin{cases}
        \overline{x_j}=x_j,\\
        \overline{y_j}=y_j,\\
        \overline{z_j}=\overline{\gamma}z_j,
    \end{cases}
\end{equation*}
the kernel function $\bm{g}$ can be expressed component-wise as
\begin{equation*}
\begin{split}
    (g_x,g_y,g_z) 
    &= \dfrac{\bigl(\overline{x}-\overline{x_j},\overline{y}-\overline{y_j},\tfrac{1}{\overline{\gamma}}(\overline{z}-\overline{z_j})\bigr)}{\left((\overline{x}-\overline{x_j})^2+(\overline{y}-\overline{y_j})^2+(\overline{z}-\overline{z_j})^2\right)^{3/2}}\\
    &=(\partial_{\overline{x_j}}, \partial_{\overline{y_j}}, \tfrac{1}{\overline{\gamma}}\partial_{\overline{z_j}})\phi\quad\text{with}\quad \phi:=\dfrac{1}{\Vert \overline{\bm{x}}-\overline{\bm{x}_j}\Vert_{2}}.
\end{split}    
\end{equation*}
Therefore, the norm of the derivative of $\bm{g}$ along each coordinate direction is 
\begin{equation*}
\begin{split}
    &\Vert\partial^{n+1}_{x_j} \bm{g}\Vert_{\infty} 
        =\max(|\partial^{n+1}_{\overline{x_j}}\partial_{\overline{x_j}}\phi|, |\partial^{n+1}_{\overline{x_j}}\partial_{\overline{y_j}}\phi|, |\partial^{n+1}_{\overline{x_j}}\partial_{\overline{z_j}}\phi| ), \\
    &\Vert\partial^{n+1}_{y_j} \bm{g}\Vert_{\infty} 
        =\max(|\partial^{n+1}_{\overline{y_j}}\partial_{\overline{x_j}}\phi|, |\partial^{n+1}_{\overline{y_j}}\partial_{\overline{y_j}}\phi|, |\partial^{n+1}_{\overline{y_j}}\partial_{\overline{z_j}}\phi| ), \\
    &\Vert\partial^{n+1}_{z_j} \bm{g}\Vert_{\infty} 
        =\overline{\gamma}^{n+1}\cdot\max(|\partial^{n+1}_{\overline{z_j}}\partial_{\overline{x_j}}\phi|,
        |\partial^{n+1}_{\overline{z_j}}\partial_{\overline{y_j}}\phi|, \tfrac{1}{\overline{\gamma}}|\partial^{n+1}_{\overline{z_j}}\partial_{\overline{z_j}}\phi| ).
\end{split}
\end{equation*}
Using \hyperref[thm:norm_kernel_derivative]{Theorem~\ref*{thm:norm_kernel_derivative}} with $a=1$ and  $\nu=n+1$, the upper bounds of the norms of the derivatives of $\bm{g}$ are
\begin{equation*}
\begin{split}
    &\Vert\partial^{n+1}_{x_j} \bm{g}\Vert_{\infty} \leq 
        \text{const}\cdot(n+1)!\cdot\Vert \overline{\bm{x}}-\overline{\bm{x}_j}\Vert^{-(n+3)}_2,\\
    &\Vert\partial^{n+1}_{y_j} \bm{g}\Vert_{\infty} \leq 
        \text{const}\cdot(n+1)!\cdot\Vert \overline{\bm{x}}-\overline{\bm{x}_j}\Vert^{-(n+3)}_2,\\
    &\Vert\partial^{n+1}_{z_j} \bm{g}\Vert_{\infty} \leq
        \text{const}\cdot(n+1)!\cdot\overline{\gamma}^{n+1}\cdot\Vert \overline{\bm{x}}-\overline{\bm{x}_j}\Vert^{-(n+3)}_2.\\
\end{split}
\end{equation*}
After the substitution of the upper bounds above into ~\eqref{eq:interpolation_error}, the interpolation error becomes
\begin{equation*}
    \begin{split}
    \Vert \bm{g}(\bm{x}, \cdot) - \widetilde{\bm{g}}(\bm{x}, \cdot) \Vert_{\infty,Q_{S}}
    &\leq \text{const}\cdot\sum_{k\in\{x,y,z\}} s^{n+1}_{k}(b_k - a_k)^{n+1}\cdot\dfrac{1}{\Vert                          \overline{\bm{x}}-\overline{\bm{x}_j}\Vert^{n+3}_{\infty}}\\
    &= \text{const}\cdot\sum_{k\in\{x,y,z\}} (\overline{b}_k -            \overline{a}_k)^{n+1}\cdot\dfrac{1}{\Vert                          \overline{\bm{x}}-\overline{\bm{x}_j}\Vert^{n+3}_{\infty}}\\
    &\leq \dfrac{\text{const}}{\text{dist}(\overline{x},\overline{S})^2}\dfrac{\text{diam}(\overline{S})^{n+1}}{\text{dist}(\overline{\bm{x}},\overline{S})^{n+1}}
    \end{split}
\end{equation*}
with a three dimensional stretch factor $\bm{s}=(1,1,\gamma)$ and $\overline{S}:=\{\bm{s}\circ\bm{x}_j\,|\,\forall \bm{x}_j\in S\}$ the set of stretched positions of particle from $S$. Here, the symbol $\circ:\mathbb{R}^n\times\mathbb{R}^n\to\mathbb{R}^n$ denotes the component-wise product of two vectors. In the derivation above, the equivalence of the two norm and the infinity norm has been used~\cite{conrad2018equivalence}.

To express the admissibility condition in the original variables, we define the stretched diameter and the stretched distance 
\begin{equation*}
\begin{split}
    \overset{\scriptscriptstyle{(s_x,s_y,s_z)}}{\text{diam}}(S)
        &:=\max_{\bm{x}_i,\bm{x}_j\in S} \Vert(s_x,s_y,s_z)\circ(\bm{x}_i-\bm{x}_j)\Vert_2,\\
    \overset{\scriptscriptstyle{(s_x,s_y,s_z)}}{\text{dist}}(\bm{x},S)
        &:=\min_{\bm{x}_j\in S} \Vert  (s_x,s_y,s_z)\circ(\bm{x}-\bm{x}_j)\Vert_2,\\
\end{split}
\end{equation*}
where $\bm{s}=(s_x,s_y,s_z)$ is a three dimensional stretch factor such that this holds:
\begin{equation*}
    \overset{\scriptscriptstyle{(1,1,\overline{\gamma})}}{\text{diam}}(S)=\text{diam}(\overline{S})
    \quad\text{and}\quad
    \overset{\scriptscriptstyle{(1,1,\overline{\gamma})}}{\text{dist}}(\bm{x},S)=\text{dist}(\overline{\bm{x}},\overline{S}).
\end{equation*}
Therefore, the stretched admissibility condition for the relativistic kernel can be defined with
\begin{equation}\label{eq:stretched_admissibility}
    \dfrac{
        \overset{\scriptscriptstyle{(1,1,\overline{\gamma})}}{\text{diam}}(S)}{\overset{\scriptscriptstyle{(1,1,\overline{\gamma})}}{\text{dist}}(\bm{x},S)}\leq \eta.
\end{equation}

\section{An interpolation-based treecode for evaluating the relativistic space-charge field}
When the distance from the particle cluster to the target point $\text{dist}(x,S)$ fulfills an admissibility condition, the separable approximation to the kernel function through the Lagrangian interpolation can be applied and the space-charge field of the particle cluster $S$ can be computed approximately by
\begin{align}\label{eqs:cluster_field}
\begin{split}
    \sum_{j\in\widehat{S}}\bm{E}(\bm{x}, \bm{x}_j)&\approx 
    \dfrac{q}{4\pi\epsilon_0}\sum_{\bm{\nu}}\sum_{j\in \widehat{S}}\ell_{S,\bm{\nu}}(\bm{x}_j)\gamma_{j}\bm{g}(\bm{x}, \bm{\xi}_{\bm{\nu}})=\dfrac{q}{4\pi\epsilon_0}\sum_{\bm{\nu}}\gamma_{S,\bm{\nu}}\cdot \bm{g}(\bm{x}, \bm{\xi}_{\bm{\nu}}),\\
    \sum_{j\in\widehat{S}}\bm{B}(\bm{x}, \bm{x}_j)&\approx 
    \dfrac{q}{4\pi\epsilon_0 c_0}
    \sum_{\bm{\nu}}\sum_{j\in \widehat{S}}\ell_{S,\bm{\nu}}(\bm{x}_j)\bm{p}_{j}\times\bm{g}(\bm{x}, \bm{\xi}_{\bm{\nu}})=\dfrac{q}{4\pi\epsilon_0 c_0}\sum_{\bm{\nu}}\bm{p}_{S,\bm{\nu}}\times \bm{g}(\bm{x}, \bm{\xi}_{\bm{\nu}}).
\end{split}
\end{align}
Here, the quantities
\begin{equation}\label{eq:effective_quantities_macroparticle}
    \gamma_{S,\bm{\nu}}:=\sum_{j\in \widehat{S}}\ell_{S,\bm{\nu}}(\bm{x}_j)\gamma_{j}
    \quad\text{and}\quad
    \bm{p}_{S,\bm{\nu}}:=\sum_{j\in \widehat{S}}\ell_{S,\bm{\nu}}(\bm{x}_j)\bm{p}_{j}
\end{equation}
are the effective Lorentz factor and the effective momentum of the macro particle with a position $\bm{\xi}_{\bm{\nu}}$ in the cluster $S$. 

Based on~\eqref{eqs:cluster_field}, we can formulate a treecode for computing the relativistic space-charge field.
When the treecode is applied with the stretched admissibility condition, it is advantageous to adjust the procedure to determine the splitting direction. As implied by the stretched diameter in~\eqref{eq:stretched_admissibility}, the splitting coordinate direction should be determined by the stretched bounding box. Therefore, in the implementation of the treecode with the stretched admissibility condition, a modified strategy to determine the splitting direction is used (\hyperref[alg:direction4split_stretching]{Algorithm~\ref*{alg:direction4split_stretching}} with $\bm{s}=(1,1,\overline{\gamma})$).

The formulated treecode (\hyperref[alg:treecode]{Algorithm~\ref*{alg:treecode}}) consists of the two main procedures:
\begin{compactenum}
    \item A given root particle cluster is subdivided into a hierarchy of smaller clusters (\hyperref[alg:cluster_subdivision]{Algorithm~\ref*{alg:cluster_subdivision}}). The splitting direction is determined by the biggest length of the bounding box (\hyperref[alg:direction4split_stretching]{Algorithm~\ref*{alg:direction4split_stretching}} with $\bm{s}=(1,1,1)$) or the stretched bounding box (\hyperref[alg:direction4split_stretching]{Algorithm~\ref*{alg:direction4split_stretching}} with $\bm{s}=(1,1,\overline{\gamma})$). 
    \item In the evaluation of the field on each target point (\hyperref[alg:cluster2p]{Algorithm~\ref*{alg:cluster2p}}), an independent traversal of the cluster tree is performed to determine the candidates of the interaction cluster. When a target point to a non-leaf cluster fulfills the conventional admissibility condition (\hyperref[alg:stretched_admissible]{Algorithm~\ref*{alg:stretched_admissible}} with $\bm{s}=(1,1,1)$) or the stretched admissibility condition (\hyperref[alg:stretched_admissible]{Algorithm~\ref*{alg:stretched_admissible}} with $\bm{s}=(1,1,\overline{\gamma})$), the effective quantities are used to compute the space-charge field~\eqref{eqs:cluster_field}. Note that, instead of \eqref{eq:definition_diam_dist}, we adapt a different definition to compute the diameter and the distance for the admissibility condition in \hyperref[alg:stretched_admissible]{Algorithm~\ref*{alg:stretched_admissible}} because of its simplicity in the practical implementation. 
\end{compactenum}

To illustrate the effect of the stretch on the particle-clustering of treecode, we performed 2D simulations for particles uniformly randomly distributed over the unit square $[0,1]^2$. The result is demonstrated in \hyperref[fig:clustering]{Figure~\ref*{fig:clustering}}.
\begin{figure}[H]
\centering
\begin{subfigure}[b]{0.49\textwidth}
    \includegraphics[width=\linewidth]{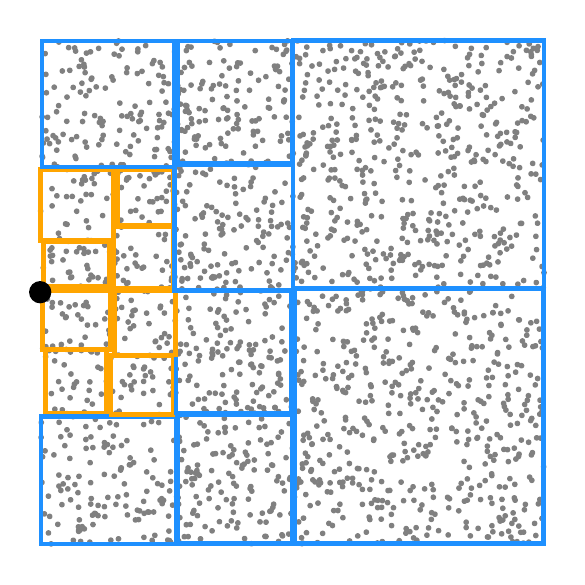}
    \caption{\label{fig:clustering_a}non-stretched}
\end{subfigure}
\begin{subfigure}[b]{0.49\textwidth}
    \includegraphics[width=\linewidth]{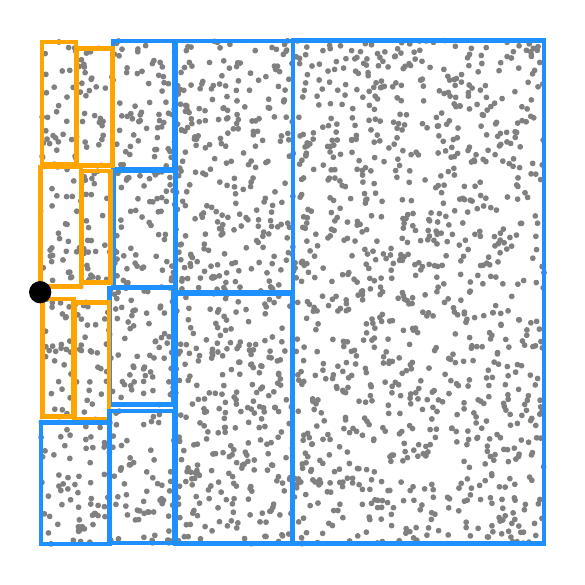}
    \caption{\label{fig:clustering_b}stretched}
\end{subfigure}
\caption{Clustering of 2D particles (gray dots) with respect to a target point (black dot) by using (a) treecode without stretch and (b) treecode with a stretch by a factor of $2.5$ in the horizontal direction (\emph{i.e.,} $\bm{s}=(2.5,1)$). The bounding box of the leaf clusters and the admissible clusters (from different levels) are marked with orange and blue, respectively. The simulation is performed with admissibility parameter $\eta=0.5$, total number of particles $N=2048$ and maximum number of particles in the leaf clusters $N_{0}=32$.}
\label{fig:clustering}
\end{figure}

\begin{algorithm}\caption{Treecode with Stretch}\label{alg:treecode}
$S$: particle cluster\\
$\bm{s}$: stretch factor\\
$N_0$: maximum number of particles in the leaf
node\\
$\eta$: admissibility parameter\\
$\bm{E}_i, \bm{B}_i$: the space-charge field experienced by the $i$-th particle\\
\SetKwProg{Func}{Function}{}{end}
\Func{\normalfont{TreecodeStretch}$(S,\,\bm{s},\,N_{0},\,\eta)$}{
    \normalfont{subdivide}$(S,\,\bm{s},\,N_{0})$ (\hyperref[alg:cluster_subdivision]{Algorithm~\ref*{alg:cluster_subdivision}})\\
    $\bm{E}_i=\bm{0}, \bm{B}_i=\bm{0}\quad \forall i\in \widehat{S}$ \\
    \For{$i\in \widehat{S}$}{
        $\bm{E}_i, \bm{B}_i = \text{cluster2p}(\bm{x}_i,\,S,\,\bm{s},\,\eta)$ (\hyperref[alg:cluster2p]{Algorithm~\ref*{alg:cluster2p}})\\
    }
    \Return $\{\bm{E}_i \mid i\in\widehat{S}\}$, $\{\bm{B}_i \mid i\in\widehat{S}\}$
}
\end{algorithm}

\begin{algorithm}\caption{Subdivision of Particle Cluster}\label{alg:cluster_subdivision}
$\gamma_{S,\bm{\nu}}$: effective Lorenz factor for cluster (global variable)~\eqref{eq:effective_quantities_macroparticle}\\ $\bm{p}_{S,\bm{\nu}}$: effective momentum for cluster (global variable)~\eqref{eq:effective_quantities_macroparticle}\\
$S$: particle cluster\\
$\bm{s}$: stretch factor\\
$N_0$: maximum number of particles in the leaf node\\
\SetKwProg{Func}{Function}{}{end}
\Func{\normalfont{subdivide}$(S,\,\bm{s},\,N_0)$}{
  \eIf{$|S|>N_0$}{
    compute $\gamma_{S,\bm{\nu}}$ and $\bm{p}_{S,\bm{\nu}}$ using ~\eqref{eq:effective_quantities_macroparticle}\\
    $k=\text{direction4split}(S,\bm{s})$ (\hyperref[alg:direction4split_stretching]{Algorithm~\ref*{alg:direction4split_stretching}})\\
    $k_{\text{split}}=$ $\lfloor |S|/2\rfloor$-th largest element of $\{k_i\mid i\in\widehat{S}\}$\\
    $S_1=\{\bm{x}_i \mid k_i\leq k_{\text{split}}, i\in\widehat{S}\}$\\
    $S_2=\{\bm{x}_i \mid k_i> k_{\text{split}}, i\in\widehat{S}\}$\\
    children(S) $= \{S_1,\,S_2\}$\\
    subdivide($S_1,\,\bm{s},\,N_0$)\\
    subdivide($S_2,\,\bm{s},\,N_0$)\\
  }{
    children(S) $= \emptyset$
  }
}
\end{algorithm}

\begin{algorithm}\caption{Direction for the Split of Particle Cluster with Stretch}\label{alg:direction4split_stretching}
$S$: particle cluster\\
$\bm{s}$: stretch factor\\
\SetKwProg{Func}{Function}{}{end}
\Func{\normalfont{direction4split}$(S,\,\bm{s})$}{
    $(s_x,s_y,s_z)=\bm{s}$\\
    $[a_x,b_x]\times[a_y,b_y]\times[a_z,b_z]=\text{bbox}(S)$\\
    $k = \underset{i\in\{x,y,z\}}{\mathrm{argmax}}\,\,s_i(b_i - a_i)$\\
    \Return $k$
}
\end{algorithm}

\begin{algorithm}\caption{Cluster Field Evaluation}\label{alg:cluster2p}
$\gamma_{S,\bm{\nu}}$: effective Lorenz factor for cluster (global variable)~\eqref{eq:effective_quantities_macroparticle}\\ $\bm{p}_{S,\bm{\nu}}$: effective momentum for cluster (global variable)~\eqref{eq:effective_quantities_macroparticle}\\
$\bm{x}$: position of the target point \\
$S$: particle cluster\\
$\bm{s}$: stretch factor\\
$\eta$: admissibility parameter\\
\SetKwProg{Func}{Function}{}{end}
\Func{\normalfont{cluster2p}$(\bm{x},\,S,\,\bm{s},\,\eta)$}{
  $\bm{E}=\bm{0},\,\bm{B}=\bm{0}$\\
  $\text{isAdmissible}=\text{admissible}(\bm{x},S,\,\bm{s},\,\eta)$ (\hyperref[alg:stretched_admissible]{Algorithm~\ref*{alg:stretched_admissible}})\\
  \uIf{\normalfont{children}$(S) == \emptyset$}{
    \For{$j \in \widehat{S}$}{
      $\bm{E}=\bm{E} + \gamma_j \bm{k}(\bm{x},\bm{x}_j,\bm{p}_j)$\\
      $\bm{B}=\bm{E} + \bm{p}_j \times \bm{k}(\bm{x},\bm{x}_j,\bm{p}_j)$\\
    }
  }
  \uElseIf{\text{\normalfont{isAdmissible}}}{
    $\bm{E}=\bm{E} + \sum\limits_{\bm{\nu}}\gamma_{S,\bm{\nu}}\cdot\bm{g}(\bm{x},\bm{\xi}_{\bm{\nu}})$ \\
    $\bm{B}=\bm{B} + \sum\limits_{\bm{\nu}}\bm{p}_{S,\bm{\nu}}\times\bm{g}(\bm{x},\bm{\xi}_{\bm{\nu}})$ \\
  }
  \Else{
    \For{$s \in \text{\normalfont{children}}(S)$}{
      $\bm{E}_{\text{temp}}$, $\bm{B}_{\text{temp}}=\text{cluster2p}(\bm{x},\,S,\,\bm{s},\,\eta)$\\
      $\bm{E}=\bm{E} + \bm{E}_{\text{temp}}$ \\
      $\bm{B}=\bm{B} + \bm{B}_{\text{temp}}$ \\
    } 
  }
  \Return $\bm{E}$, $\bm{B}$
}
\end{algorithm}

\begin{algorithm}\caption{Stretched Admissibility Condition}\label{alg:stretched_admissible}
$\bm{x}$: position of the target point\\
$S$:  particle cluster\\
$\bm{s}$: stretch factor\\
$\eta$: admissibility parameter\\
\SetKwProg{Func}{Function}{}{end}
\Func{\normalfont{admissible}$(\bm{x},\,S,\bm{s},\eta)$}{
   $[a_x,b_x]\times[a_y,b_y]\times[a_z,b_z]=\text{bbox}(S)$\\
   $\bm{c} = (\bm{a} + \bm{b})/2$\\
   $\text{diam} = \Vert \bm{s}\circ(\bm{b}-\bm{c})\Vert_2$ \\
   $\text{dist} = \Vert \bm{s}\circ(\bm{x}-\bm{c})\Vert_2$ \\
   \Return $\text{diam}/\text{dist} \leq \eta$
}
\end{algorithm}

\section{An alternative method based on relativity transformation}
In the previous section we derived an upper bound of the interpolation error for the relativistic interaction kernel and introduced a stretched admissibility condition. In this section we introduce an alternative approach based on the relativity transformation. With this approach, the treecode can be performed without switching to a stretched admissibility condition. We know that the necessity of the stretched admissibility condition comes from the average particle momentum term $\overline{\bm{p}}$ in the denominator of the relativistic kernel~\eqref{eq:kernel_g}. If some transformation can be applied to eliminate the average momentum term, we can just use treecode with the conventional admissibility condition. The idea is to transform every physical quantity to the average rest-frame (AVGRF, the inertial frame moving with particle beam's average momentum) through the Lorentz transformation and evaluate the particle interaction in the frame. In the average rest-frame, the interaction kernel is approximately equal to the electrostatic kernel and treecode can be directly performed with the conventional admissibility condition. 

\subsection{Particle interaction in the average rest-frame}
Let $\overline{\bm{p}}$ be the average momentum of particles in the system in the lab-frame $\mathcal{K}$ and $\mathcal{K}'$ be the average rest-frame, in which the average particle momentum is $\overline{\bm{p}}'=0$. According to \eqref{eq:space_time_transformation} and \eqref{eq:4momentum_transformation} with $\bm{p}_{u}=\overline{\bm{p}}$, the position and the momentum in $\mathcal{K}$ and $\mathcal{K}'$ have the relations
\begin{align}
  \label{eq:position_average_frame}
  x'_{\parallel} &= \overline{\gamma}x_{\parallel} -\overline{p}(c_0t), &
  x'_{\perp} &= x_{\perp},\\
  \label{eq:momentum_average_frame}
  p'_{\parallel} &= \overline{\gamma}p_{\parallel} - \overline{p}\gamma, &
  p'_{\perp} &= p_{\perp},
\end{align}
where $\parallel$ and $\perp$ denote the components parallel and perpendicular to $\overline{\bm{p}}$. Here, we use $\overline{p}$ to denote the magnitude of $\overline{\bm{p}}$ (\emph{i.e,} $\overline{p}:=\Vert\overline{\bm{p}}\Vert_2$). 

The momentum of a particle can be written as $\bm{p} = \overline{\bm{p}} + \Delta\bm{p}$ and can also be component-wise expressed in the plane spanned by the average momentum $\overline{\bm{p}}$ and the momentum $\bm{p}$ of the particle as
\begin{equation}\label{eq:relative_momentum}
    p_{\parallel} = \overline{p} +\Delta p_{\parallel}
    \quad\text{and}\quad
    p_{\perp} = \Delta p_{\perp},
\end{equation}
where $\Delta \bm{p}$ is the deviation of the particle momentum from $\overline{\bm{p}}$ with the scalar components $\Delta p_{\parallel}$, $\Delta p_{\perp}$. Throughout this study, we assume that $\overline{p} \gg \Delta p_{\parallel}, \Delta p_{\perp}$, which is valid in the physics of particle beams.

To eliminate the momentum appearing in the denominator of the relativistic kernel, we can evaluate the particle field in $\mathcal{K}'$,
\begin{align}\label{eq:emfield_in_average_frame}
    \bm{E}'&=\dfrac{1}{4\pi\epsilon_0}\dfrac{\gamma'\bm{R}'}{\bigl| \Vert\bm{R}'\Vert^2_2 + \left(\bm{p}'\cdot\bm{R}'\right)^2\bigr|^{3/2}}, &
    \bm{B}'&=\dfrac{1}{4\pi\epsilon_0 c_0}\dfrac{\bm{p}_j\times\bm{R}'}{\bigl| \Vert\bm{R}'\Vert^2_2 + \left(\bm{p}'\cdot\bm{R}'\right)^2\bigr|^{3/2}},
\end{align}
where $\bm{R}'$ is the \replaced[id=R3-1]{vector}{distance} between source and target particle, $\gamma'$ and $\bm{p}_j$ are the Lorentz factor and the momentum of the source particle, respectively. Let $\Delta\gamma$ denote the deviation to the average Lorentz factor $\overline{\gamma}$ of the particle, $\gamma=\overline{\gamma}+\Delta\gamma$. By~\eqref{eq:momentum_average_frame} and \eqref{eq:relative_momentum},
the particle momentum in the average rest-frame can be approximated by
\begin{equation}\label{eq:approx_rest_frame_momentum}
    p'_{\parallel}
    =\overline{\gamma} (\overline{p}+\Delta p_{\parallel}) - \overline{p}(\overline{\gamma}+\Delta \gamma)
    =\overline{\gamma}\Delta p_{\parallel} - \overline{p}\Delta\gamma
    \approx \dfrac{\overline{\gamma}^2-\overline{p}^2}{\overline{\gamma}}\Delta p_{\parallel}
    =\dfrac{1}{\overline{\gamma}}\Delta p_{\parallel}, \quad
    p'_{\perp} = \Delta p_{\perp},
\end{equation}
where $\Delta\gamma$ is estimated using the identites $\gamma^2 = 1+\bm{p}\cdot\bm{p}$ and $\overline{\gamma}^2=1+\overline{p}^2$:
\begin{alignat*}{2}
    &&(\overline{\gamma}+\Delta\gamma)^2 &= 1 + (p_{\parallel}+\Delta p_{\parallel})^2 + (\Delta p_{\perp})^2 \\
    \Longrightarrow\quad&&\overline{\gamma}^2+2\overline{\gamma}\Delta\gamma + (\Delta\gamma)^2 &= 1 + \overline{p}^2+2\overline{p}\Delta p_{\parallel} + (\Delta p_{\parallel})^2 + (\Delta p_{\perp})^2\\
    \Longrightarrow\quad&&\overline{\gamma}^2+2\overline{\gamma}\Delta\gamma &\approx 1 + \overline{p}^2+2\overline{p}\Delta p_{\parallel}\\
    \Longrightarrow\quad&&\Delta\gamma&\approx\dfrac{\overline{p}}{\overline{\gamma}}\Delta p_{\parallel}.
\end{alignat*}
Substituting~\eqref{eq:approx_rest_frame_momentum} into~\eqref{eq:emfield_in_average_frame}, the denominator in \eqref{eq:emfield_in_average_frame} can be approximated by
\begin{align*}
    \bigl| \Vert\bm{R}'\Vert^2_2 + \left(\bm{p}'\cdot\bm{R}'\right)^2
    \bigr|^{3/2}
    \approx 
        \bigl|R'^{2}_{\parallel}+R'^2_{\perp} + (\dfrac{1}{\overline{\gamma}}\Delta p_{\parallel}R'_{\parallel} + \Delta p_{\perp}R'_{\perp})^2\bigr|^{3/2}
    \approx \Vert\bm{R}'\Vert^3_2,
\end{align*}
where the assumption $\frac{\Delta p_{\parallel}}{\overline{\gamma}},\,\Delta p_{\perp}\ll 1$ was used. Therefore, the particle field in $\mathcal{K}'$ is approximately equal to
\begin{equation}\label{eq:emfield_in_average_frame_approx}
    \bm{E}' \approx\dfrac{1}{4\pi\epsilon_0}\dfrac{\gamma'\bm{R}'}{\Vert\bm{R}'\Vert^3_2},\quad
    \bm{B}' \approx\dfrac{1}{4\pi\epsilon_0 c_0}\dfrac{\bm{p}_j\times\bm{R}'}{\Vert\bm{R}'\Vert^3_2}.
\end{equation}
After $\bm{E}'$ and $\bm{B}'$ is computed by the treecode, the corresponding particle field in $\mathcal{K}$ can be calculated with the field transformation~\cite{jackson1999classical}
\begin{equation}\label{eq:emfield_to_lab-frame}
    \bm{E}=\overline{\gamma}\bm{E}'-\overline{\bm{p}}\times\bm{B}'-\dfrac{1}{\overline{\gamma}+1}(\overline{\bm{p}}\cdot\bm{E}')\overline{\bm{p}},\quad
    \bm{B}=\overline{\gamma}\bm{B}'+\overline{\bm{p}}\times\bm{E}'-\dfrac{1}{\overline{\gamma}+1}(\overline{\bm{p}}\cdot\bm{B}')\overline{\bm{p}}.
\end{equation}
The procedure discussed above is summarized in \hyperref[alg:treecode_AVGRF]{Algorithm~\ref*{alg:treecode_AVGRF}} and a schematic of the procedure is illustrated in \hyperref[fig:treecode-avgrf_schematic]{Figure~\ref*{fig:treecode-avgrf_schematic}}. In the practical implementation, instead of~\eqref{eq:position_average_frame}, we use
\begin{equation}\label{eq:position_average_frame_practice}
    x'_{\parallel}=\overline{\gamma}x_{\parallel}
    \quad\text{and}\quad
    x'_{\perp} = x_{\perp}
\end{equation}
to compute the particle position in the average rest-frame. Because the lab-frame position of particles is given at a common simulation time $t$, both~\eqref{eq:position_average_frame} and~\eqref{eq:position_average_frame_practice} lead to the same pair-wise distance between particles at the average-rest frame. Therefore, equations~\eqref{eq:position_average_frame_practice} can be applied in the computation of particle field by treecode, where only the relative position of particles is of importance.

\begin{figure}[H]
    \centering
    \includegraphics[width=0.8\linewidth]{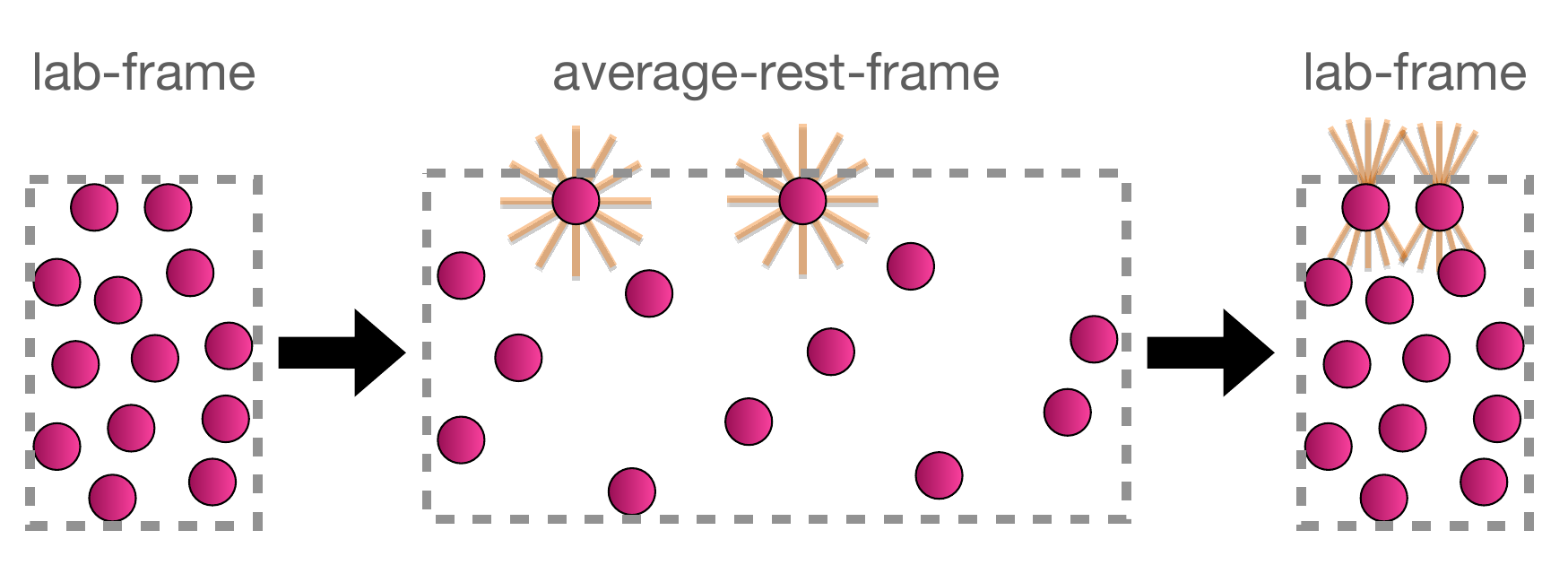}
    \caption{A schematic of the treecode based on the transformation to the average rest-frame. The particle's position and momentum are transformed to the average rest-frame and the particle field evaluated in the average rest-frame is then transformed back to the lab-frame.}
    \label{fig:treecode-avgrf_schematic}
\end{figure}

\begin{algorithm}\caption{Treecode with Average Rest-Frame Technique}\label{alg:treecode_AVGRF}
S: particle cluster\\
$\bm{s}$: stretch factor\\
$N_0$: maximum number of particles in the leaf node\\
$\eta$: admissibility parameter\\
$\bm{E}_i, \bm{B}_i$: the space-charge field experienced by the $i$-th particle\\
\SetKwProg{Func}{Function}{}{end}
\Func{\normalfont{TreecodeAVGRF}$(S,\,N_{0},\,\eta)$}{
    $\bm{s}=(1,1,1)$\\
    $S'=\text{lab2avgframe}(S)$ (\hyperref[alg:lab2avgframe]{Algorithm~\ref*{alg:lab2avgframe}})\\ 
    \normalfont{subdivide}$(S',\,\bm{s},\,N_0)$ (\hyperref[alg:cluster_subdivision]{Algorithm~\ref*{alg:cluster_subdivision}})\\
    $\bm{E}_i=\bm{0}, \bm{B}_i=\bm{0}\quad \forall i\in \widehat{S}$ \\
    \For{$i\in \hat{S'}$}{
        $\bm{E}'_{\text{temp}}, \bm{B}'_{\text{temp}} = \text{cluster2p}(\bm{x}'_i,\,S'\,\bm{s},\,\eta)$ (\hyperref[alg:cluster2p]{Algorithm~\ref*{alg:cluster2p}})\\
        compute $\bm{E}_i, \bm{B}_i$ using $\bm{E}'_{\text{temp}}, \bm{B}'_{\text{temp}}$ and \eqref{eq:emfield_to_lab-frame}
    }
    \Return $\{\bm{E}_i \mid i\in\widehat{S}\}$, $\{\bm{B}_i \mid i\in\widehat{S}\}$
}
\end{algorithm}

\begin{algorithm}\caption{Transformation of Particle Beam to the Average Rest-Frame}
\label{alg:lab2avgframe}
S: particle cluster\\
\SetKwProg{Func}{Function}{}{end}
\Func{\normalfont{lab2avgframe}$(S)$}{
    $\overline{\bm{p}}=\tfrac{1}{|S|}(\sum_{i\in\widehat{S}}\bm{p}_i)$\\
    $\overline{\gamma}=(1+\overline{\bm{p}}\cdot\overline{\bm{p}})^{1/2}$\\
    \For{$i\in\widehat{S}$}{
        $\bm{x}'_i=\bm{x}_i+\tfrac{1}{\overline{\gamma}+1}(\bm{x}_i\cdot \overline{\bm{p}})\overline{\bm{p}}$ (vector form of~\eqref{eq:position_average_frame_practice})\\
        $\bm{p}'_i=\bm{p}_i + \tfrac{1}{\overline{\gamma}+1}(\bm{p}_{i}\cdot\overline{\bm{p}})\overline{\bm{p}} - \gamma_{i}\overline{\bm{p}}$ (vector form of~\eqref{eq:momentum_average_frame})\\
        $\gamma'_{i}=(1+\bm{p}_{i}\cdot\bm{p}_{i})^{1/2}$\\
    }
    \Return $S'=\{\bm{x}'_i\, | \,\forall i \in \widehat{S} \}$
}
\end{algorithm}

\section{Results and discussion}
To understand the performance of the different treecodes discussed in this study, we performed numerical simulations for particle beams with different configurations. All the treecodes and algorithms discussed in this study have been implemented as a package \added[id=A-3]{\href{https://github.com/ykkan/NChargedBodyTreecode.jl}{NChargedBodyTreecode.jl}} in the Julia programming language~\cite{bezanson2017julia}. We refer to the procedure proposed by Wang \emph{et al.}~\cite{wang2020kernel} to compute the effective Lorentz factor and the effective momentum of the macro particle: the Lagrangian polynomials are implemented based on the second form of the barycentric formula~\cite{berrut2004barycentric}, and the interpolation points (\emph{i.e.}, the positions of the macro particles) are generated by Chebyshev points of the second kind~\cite{wang2020kernel}. \deleted[id=A-3]{The source code of our package can be accessed via~\url{https://doi.org/10.5281/zenodo.6464509}}. 

\subsection{Particle beam with single momentum}
We first consider particles randomly uniformly distributed in the unit cube $[0,1]^3$ where each particle has the same momentum $\bm{p}=(0,0,p_0)$ with $p_0=(\gamma^2-1)^{1/2}$. 
The space-charge field (both E- and B-field) experienced by each particle in the system is computed by the following methods:
\begin{compactenum}
    \item brute-force method \added[id=R3-2]{where the non-approximated kernel function $\bm{k}$ in~\eqref{eq:treecode_exact_k_kernel} is used},
    \item a treecode called “Treecode-Uniform” using the conventional admissibility condition (\hyperref[alg:stretched_admissible]{Algorithm~\ref*{alg:stretched_admissible}} with $\bm{s}=(1,1,1)$),
    \item a treecode called “Treecode-Stretch” using the stretched admissibility condition and the subdivision strategy with stretch (\hyperref[alg:stretched_admissible]{Algorithm~\ref*{alg:stretched_admissible}} with $\bm{s}=(1,1,\overline{\gamma})$),
    \item a treecode called “Treecode-AVGRF” using the average rest-frame approach (\hyperref[alg:treecode_AVGRF]{Algorithm~\ref*{alg:treecode_AVGRF}}).
\end{compactenum}
The result from brute-force method serves as a reference to evaluate the speedup and the relative error of different treecodes. The measured error is the maximal relative error in the electrical and magnetic field,
\begin{equation*}
    \text{error}:=
      \max_{\bm{f}\in\{\bm{E},\bm{B}\}}
      \biggl(\sum^{N}_{i=1} \Vert \bm{f}^{t}_{i}-\bm{f}^{b}_{i} \Vert^{2}_{2} 
      / \sum^{N}_{i=i}\Vert \bm{f}^{b}_{i} \Vert^{2}_{2}\biggr)^{1/2},
\end{equation*}
where $N$ is the number of particles in the system. The space-charge fields $\bm{f}^{t}_{i}$ and $\bm{f}^{b}_{i}$ experienced by the $i$-th particle are computed by treecode and the brute-force method, respectively. The result is shown in \hyperref[fig:monomom_errors]{Figure~\ref*{fig:monomom_errors}}. As illustrated by \hyperref[fig:monomom_errors_a]{Figure~\ref*{fig:monomom_errors_a}}, Treecode-Uniform has large errors and this supports our argument in the beginning that the conventional treecode cannot be used directly for relativistic particle beams. On the other hand, as demonstrated in \hyperref[fig:monomom_errors_b]{Figure~\ref*{fig:monomom_errors_b}} and \hyperref[fig:monomom_errors_c]{Figure~\ref*{fig:monomom_errors_c}}, the two proposed treecodes can effectively treat the problem of a relativistic particle beam. Also, we can observe that the errors computed by Treecode-Stretch and Trecode-AVGRF look the same. The reason is that both methods result in equivalent ways to subdivide the particle clusters and to determine the admissible clusters. In Treecode-AVGRF, the particles' spatial distribution gets boosted along the longitudinal direction after the relativity transformation depicted in \hyperref[fig:treecode-avgrf_schematic]{\hyperref[fig:treecode-avgrf_schematic]{Figure~\ref*{fig:treecode-avgrf_schematic}}}. In Treecode-Stretch, the particles' spatial distribution stays unchanged in the lab-frame but the subdivision and the admissibility condition involve a stretch factor related to the Lorentz boost.
\begin{figure}[H]
\centering
\begin{subfigure}[b]{0.47\textwidth}
    \includegraphics[width=\linewidth]{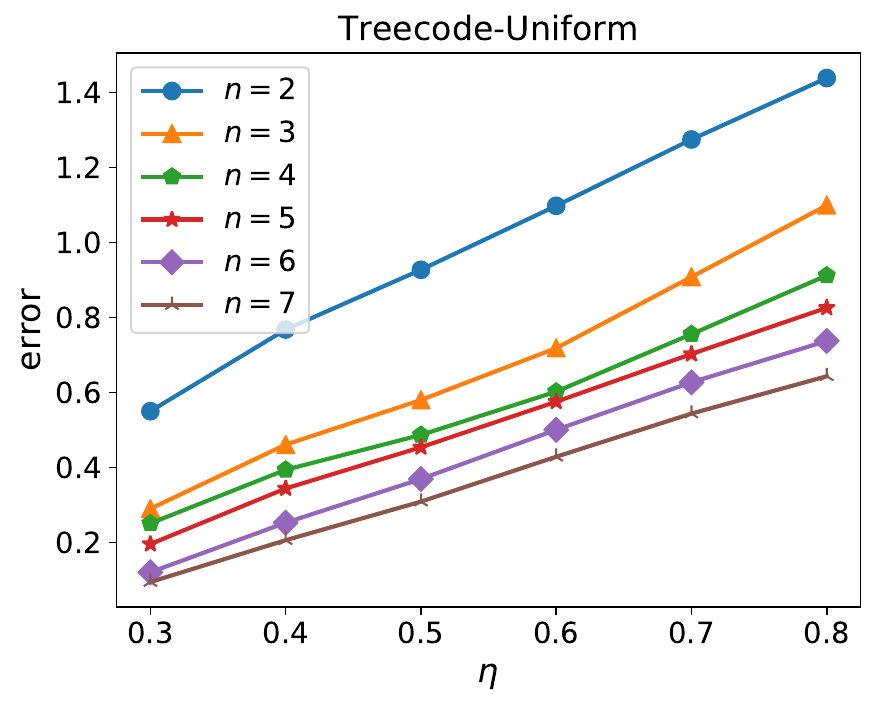}
    \caption{\label{fig:monomom_errors_a}Treecode-Uniform}
\end{subfigure}\\
\begin{subfigure}[b]{0.47\textwidth}
    \includegraphics[width=\linewidth]{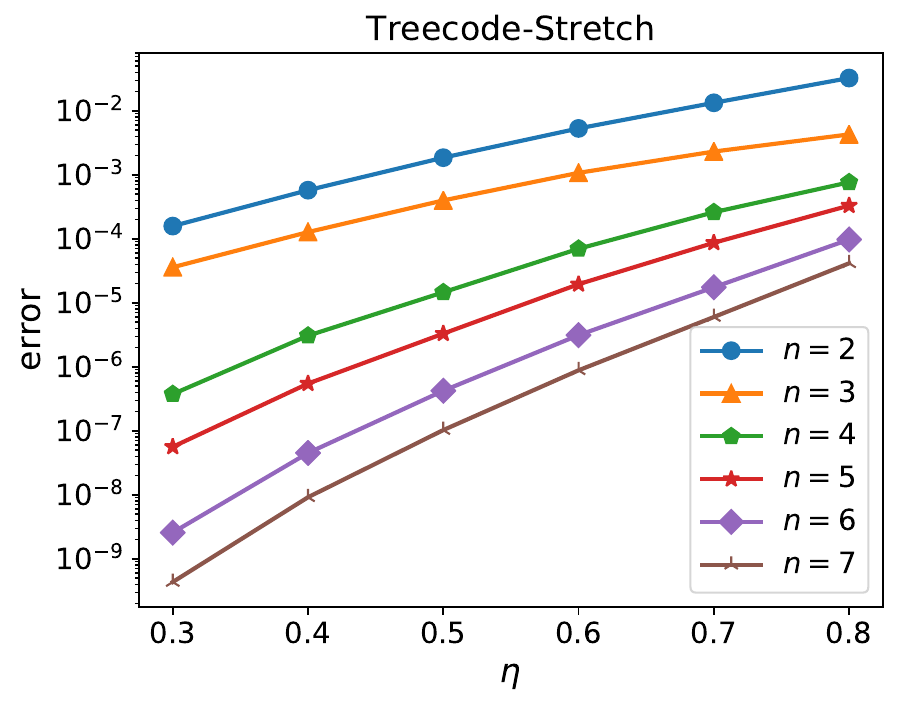}
    \caption{\label{fig:monomom_errors_b}Treecode-Stretch}
\end{subfigure}
\begin{subfigure}[b]{0.47\textwidth}
    \includegraphics[width=\linewidth]{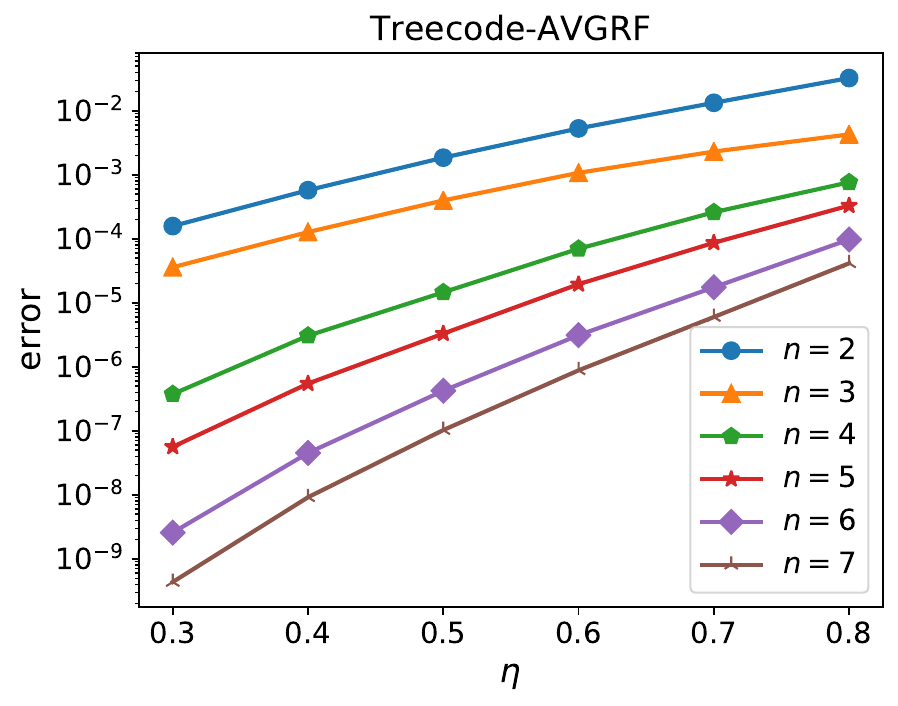}
    \caption{\label{fig:monomom_errors_c}Treecode-AVGRF}
\end{subfigure}
\caption{The accuracy of different treecodes for the particle beam with mono-momentum. The number of particles is $8\times 10^4$ and $\gamma=50$. Here, $n$ and $\eta$ are the interpolation degree and the admissibility parameter, respectively. The interpolation degree is $n=2,3,4,5,6,7$ and the maximum number of the particles in the leaf cluster is chosen as $N_0=(n+1)^3$.}
\label{fig:monomom_errors}
\end{figure}

\subsection{Particle beam with momentum spread}
In practical scenarios, beam divergence and energy spread are typically present in particle beams. These two quantities are related to the momentum distribution of a particle beam in transverse and longitudinal direction, respectively. Therefore it is also important to know how the momentum distribution influences the accuracy of treecodes. We consider a particle beam randomly distributed over a momentum distribution defined as 
\begin{align*}
    f_{\bm{p}}(p_x,p_y,p_z)\propto
    \exp\left(-\dfrac{p^2_x + p^2_y}{2(\delta_{\perp}p_0)^2}\right)
    \exp\left(-\dfrac{(p_z-p_0)^2}{2(\delta_{\parallel}p_0)^2}\right)
    \text{ with }
    p_0=(\gamma^2-1)^{1/2}.
\end{align*} 
Here, $\delta_{\perp}$ and $\delta_{\parallel}$ are equivalent to rms (root mean square) divergence and rms energy spread for a relativistic particle beam with the paraxial approximation (\emph{i.e.}, $p_{\{x,y\}}\ll p_0$). The positions of the particles are uniformly randomly distributed over the unit cube $[0,1]^3$. The result is shown in \hyperref[fig:momspread_errors]{Figure~\ref*{fig:momspread_errors}}. Our numerical results reveal that treecode-AVGRF is less accurate than Treecode-Stretch when a momentum spread is considered in the particle beam. 
\begin{figure}[H]
\centering
\begin{subfigure}[b]{0.49\textwidth}
    \includegraphics[width=\linewidth]{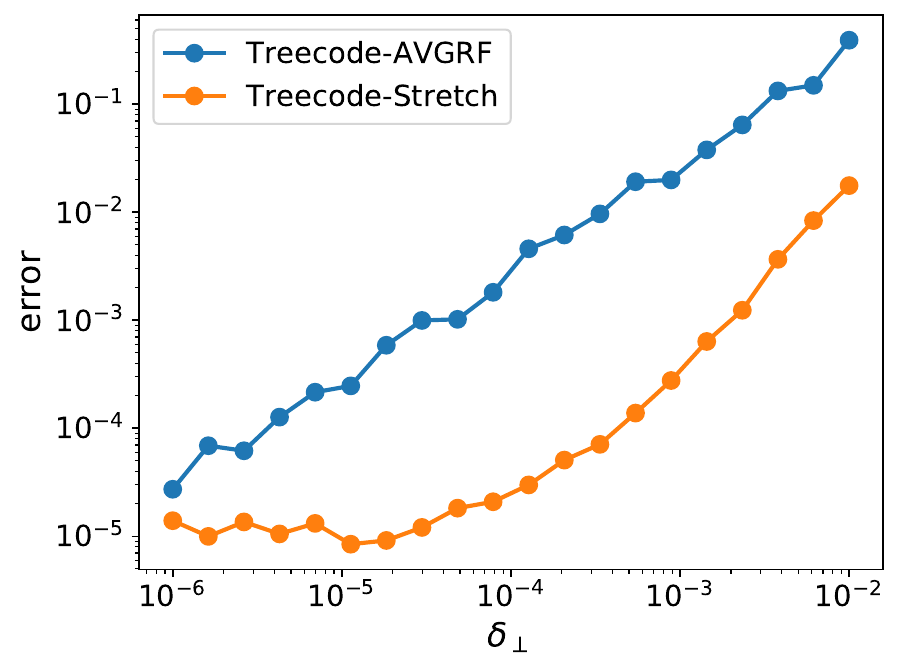}
    \caption{transverse}
\end{subfigure}
\begin{subfigure}[b]{0.49\textwidth}
    \includegraphics[width=\linewidth]{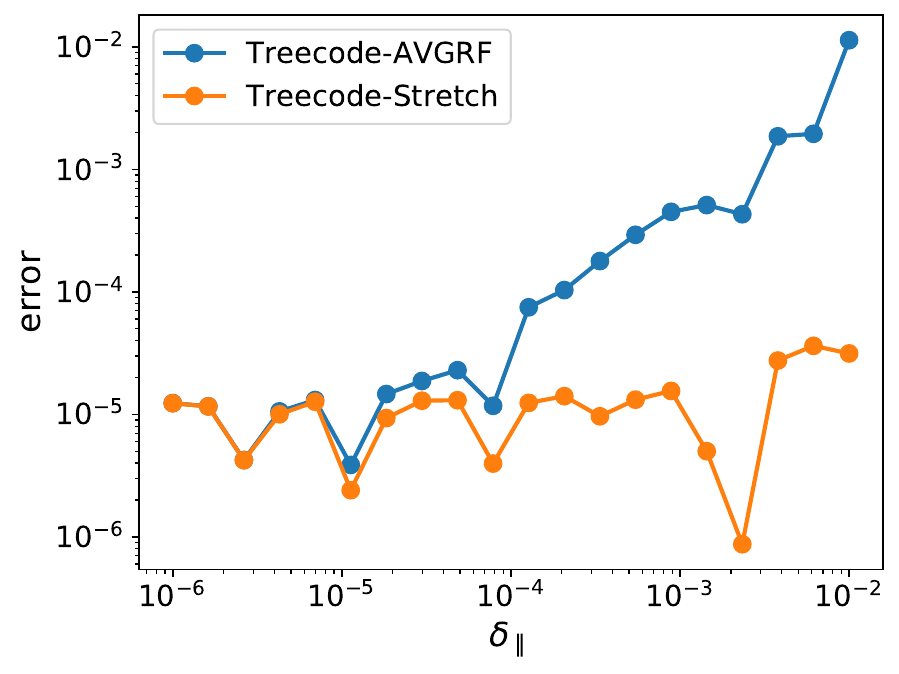}
    \caption{longitudinal}
\end{subfigure}
\caption{The accuracy of Treecode-AVGRF and Treecode-Stretch for particle beams with (a) transverse momentum spread and (b) longitudinal momentum spread. The number of particles is $8\times 10^4$ and $\gamma=50$. Each simulation is performed with admissibility parameter $\eta=0.5$, interpolation degree $n=4$, and maximum number of particles in the leaf clusters $N_0=256$.}
\label{fig:momspread_errors}
\end{figure}

\subsection{The problem of Treecode-AVGRF}
Our numerical results show that Treecode-AVGRF is less accurate compared to Treecode-Stretch when a momentum spread is considered in the particle beam. The reason is due to the inaccurate computation of the target-source distance in the average rest-frame. To explain this, we consider a particle moving with a momentum $p_j$ in the lab-frame $\mathcal{K}$ (\hyperref[fig:avgrf_problem_a]{Figure~\ref*{fig:avgrf_problem_a}}). If the positions of the particle $x_j$ and the target point $x$ are both measured at the same time $t_{x_j}=t_x = t$, the distance between these two points can be calculated directly using $x-x_j$ . If we transform the system to a reference frame $\mathcal{K}'$ moving with a momentum $\overline{p}$ (\hyperref[fig:avgrf_problem_b]{Figure~\ref*{fig:avgrf_problem_b}}), the events observed in $\mathcal{K}'$ will become
\begin{equation*}
    \begin{cases}
        x'=\overline{\gamma}x - \overline{p}t,\\
        t'_{x}=\overline{\gamma}t-\overline{p}x,\\
    \end{cases}
    \quad
    \begin{cases}
        x'_j= \overline{\gamma}x_j - \overline{p}t,\\
        t'_{x_j}=\overline{\gamma}t-\overline{p}x_{j},\\
    \end{cases}
    \quad\text{and}\qquad
    p'_j \approx \overline{p}\Delta p_j.
\end{equation*}
Here, we assume $p_j=\overline{p}+\Delta p_j$ with $\overline{p}\gg \Delta p_j$ and the space-time quantities are normalized such that the speed of light is $c_0=1$. It is problematic if we compute the distance directly by $x' - x'_j$ because these two events are measured at different times (\emph{i.e.}, $t'_{x}\neq t'_{x_j}$) and the particle keeps moving during the measurement. To get a correct distance, as demonstrated in \hyperref[fig:avgrf_problem_c]{Figure~\ref*{fig:avgrf_problem_c}}, we need to transform the system to the particle's rest-frame $\mathcal{K}''$ that satisfies
\begin{equation*}
    \begin{cases}
        x''=\gamma_j x - p_j t,\\
        t''_{x}=\gamma_j t - p_j x,\\
    \end{cases}
    \quad
    \begin{cases}
        x''_j= \gamma_j x_j - p_j t,\\
        t''_{x_j}=\gamma_j t- p_j x_{j},\\
    \end{cases}
    \quad\text{and}\qquad
    p''_j = 0.
\end{equation*}
In the reference frame $\mathcal{K}''$, although two events are still measured at different times (\emph{i.e.}, $t''_{x}\neq t''_{x_j}$), we can still compute the distance with $x'' - x''_j$ because the particle is stationary during the measurement. However, this approach is not applicable to treecode because the transformations of the whole particle beam to each particle's rest-frame have a complexity of $\mathcal{O}(N^2)$.
\begin{figure}[th]
\centering
\begin{subfigure}[b]{0.3\textwidth}
    \includegraphics[width=\linewidth]{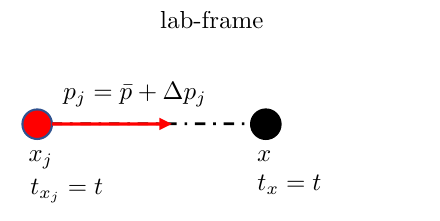}
    \caption{\label{fig:avgrf_problem_a}lab-frame $\mathcal{K}$}
\end{subfigure}
\begin{subfigure}[b]{0.3\textwidth}
    \includegraphics[width=\linewidth]{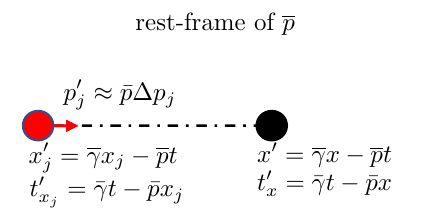}
    \caption{\label{fig:avgrf_problem_b}rest-frame $\mathcal{K}'$}
\end{subfigure}
\begin{subfigure}[b]{0.3\textwidth}
    \includegraphics[width=\linewidth]{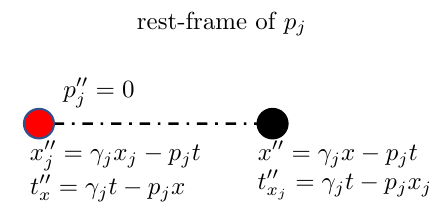}
    \caption{\label{fig:avgrf_problem_c}rest-frame $\mathcal{K}''$}
\end{subfigure}
\caption{A 1D illustration of the distance between one point (black dot) and one particle (red dot) moving with $p_j=\overline{p}+\Delta p_j$. The system is viewed in (a) the lab-frame ($\mathcal{K}$), (b) the rest-frame of $\overline{p}$ ($\mathcal{K}'$) and (c) the rest-frame of $p_j$ ($\mathcal{K}''$).}
\label{fig:avgrf_problem}
\end{figure}

\subsection{Performance}
To understand the speedup achieved with our treecode, we compare the performance of the brute-force method and the treecode in \hyperref[fig:performance]{Figure~\ref*{fig:performance}}. Here, Treeecode-Stretch is selected because it has the best accuracy in all the scenarios discussed before. \added[id=R3-3]{In this performance test, the particles are randomly uniformly distributed in the unit cube $[0,1]^3$ and each particle has the same momentum associated with $\gamma=50$. The elapsed time of treecode represents the execution time of \hyperref[alg:cluster2p]{Algorithm~\ref*{alg:treecode}}.} As shown in \hyperref[fig:performance_a]{Figure~\ref*{fig:performance_a}} and \hyperref[fig:performance_c]{Figure~\ref*{fig:performance_c}}, the brute-force method scales with $\mathcal{O}(N^2)$. The treecode scales between $\mathcal{O}(N^2)$ and $\mathcal{O}(N\log N)$ (\hyperref[fig:performance_a]{Figure~\ref*{fig:performance_a}}) and approaches to its theoretical complexity $\mathcal{O}(N\log N)$ as the number of particles $N$ becomes bigger (\hyperref[fig:performance_d]{Figure~\ref*{fig:performance_d}}). Also, we can observe that the speedup of the treecode increases with the number of particles  (\hyperref[fig:performance_b]{Figure~\ref*{fig:performance_b}}). Here, the speedup refers to the ratio of the elapsed time of the brute-force method to the elapsed time of the treecode. \added[id=R2-2]{Besides, we can observe in \hyperref[fig:performance_e]{Figure~\ref*{fig:performance_e}} that the relative error of treecode is independent of $N$.}
\begin{figure}[H]
\centering
\begin{subfigure}[b]{0.4\textwidth}
    \includegraphics[width=\linewidth]{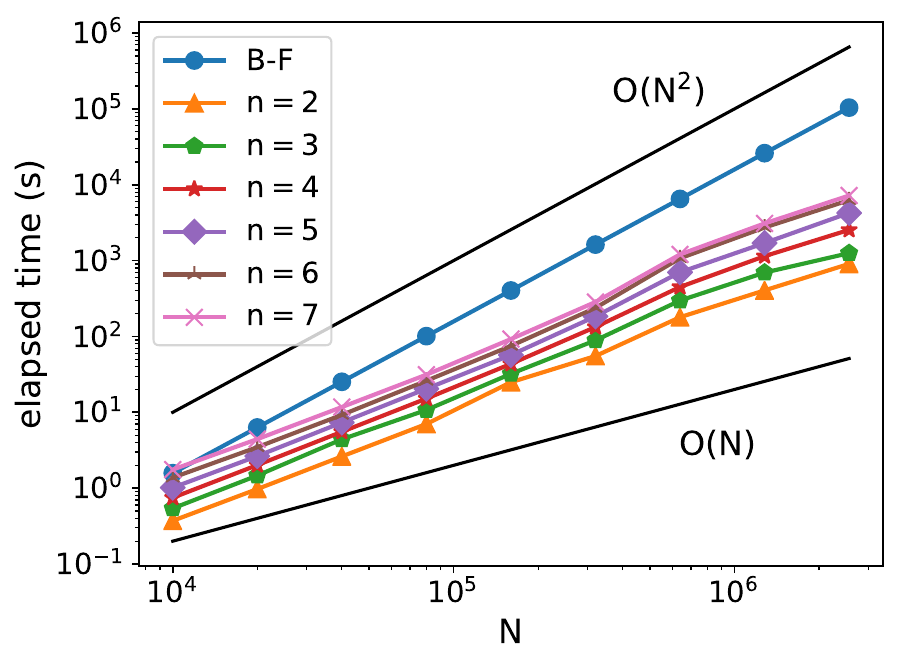}
    \caption{\label{fig:performance_a}elapsed time}
\end{subfigure}
\begin{subfigure}[b]{0.4\textwidth}
    \includegraphics[width=\linewidth]{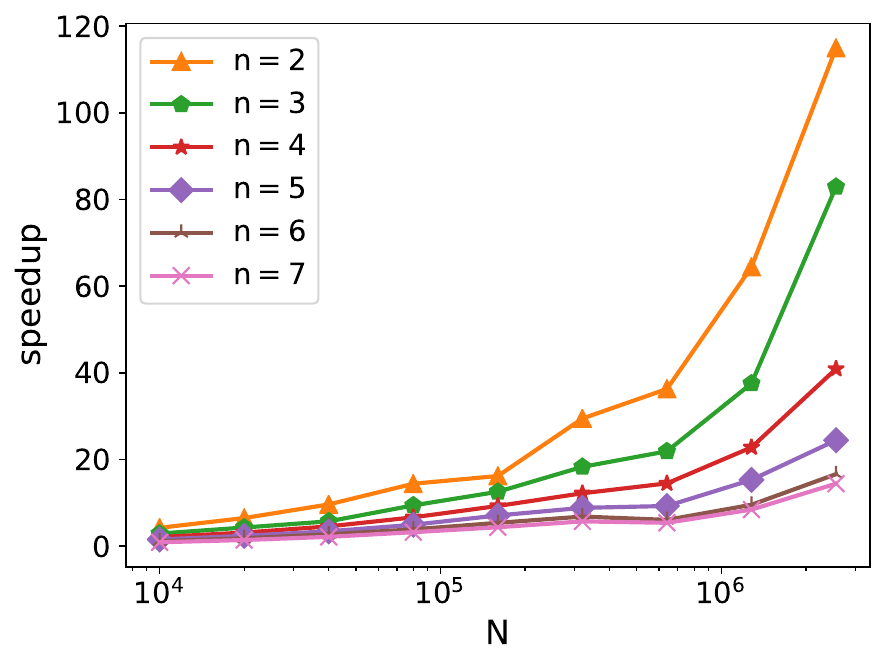}
    \caption{\label{fig:performance_b}speedup}
\end{subfigure}
\begin{subfigure}[b]{0.4\textwidth}
    \includegraphics[width=\linewidth]{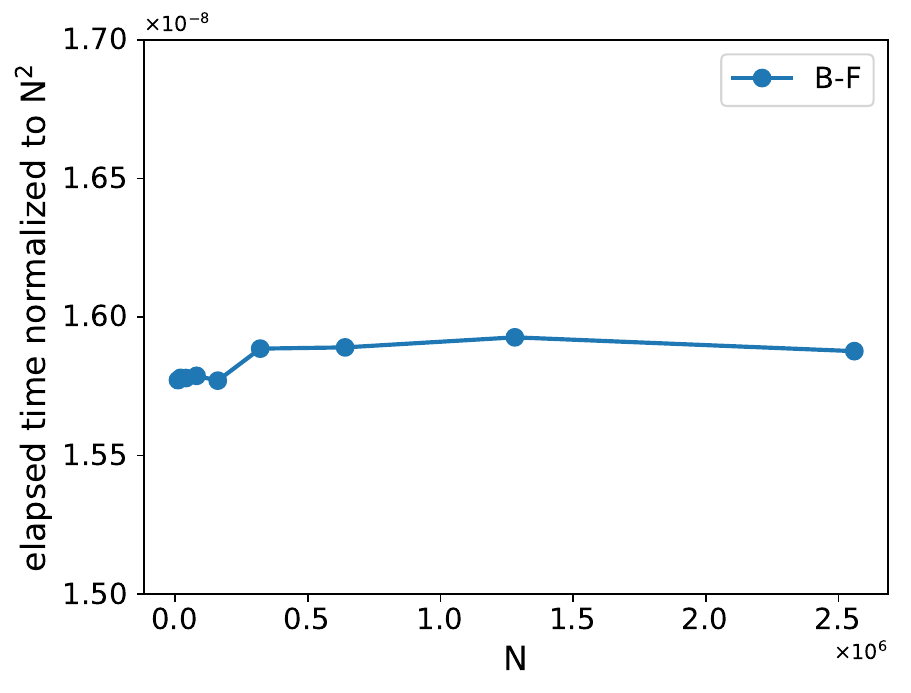}
    \caption{\label{fig:performance_c}normalized elapsed time}
\end{subfigure}
\begin{subfigure}[b]{0.4\textwidth}
    \includegraphics[width=\linewidth]{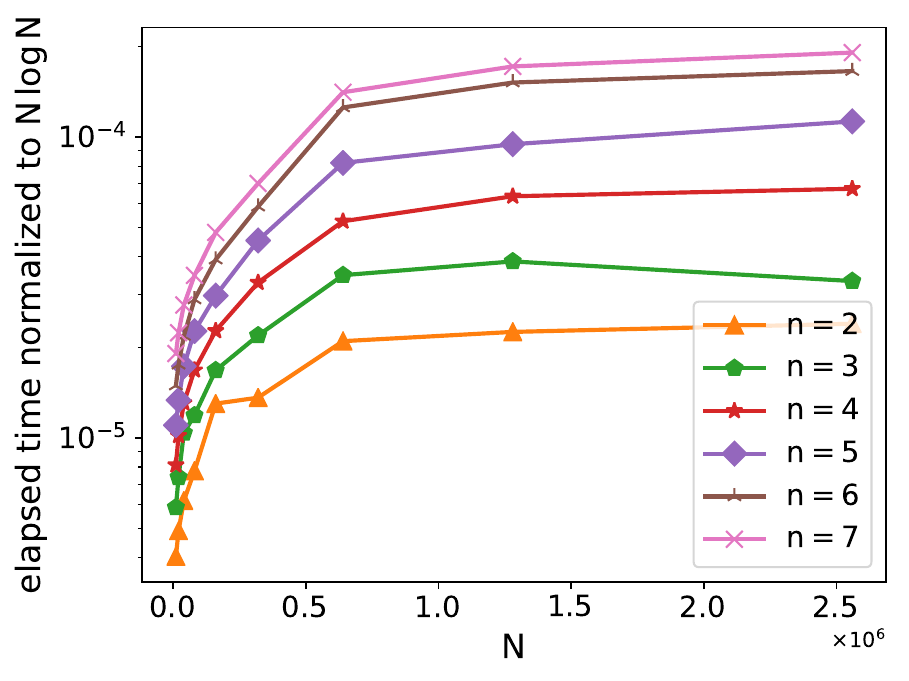}
    \caption{\label{fig:performance_d}normalized elapsed time}
\end{subfigure}
\begin{subfigure}[b]{0.4\textwidth}
    \includegraphics[width=\linewidth]{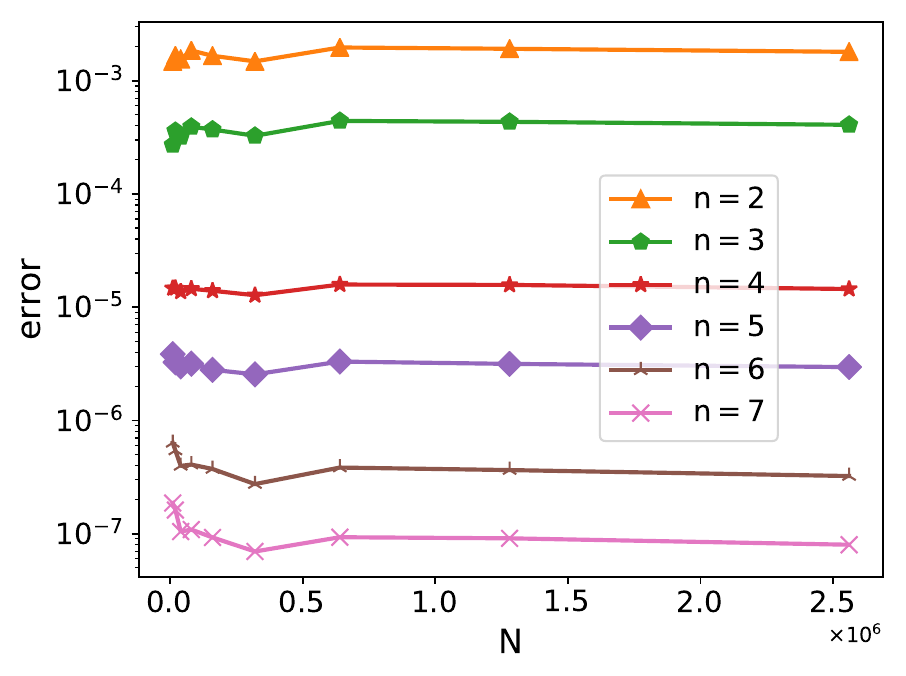}
    \caption{\label{fig:performance_e}error}
\end{subfigure}
\caption{The performance of brute-force and treecode methods. \hyperref[fig:performance_a]{Figure~\ref*{fig:performance_a}} shows the elapsed time used by brute-force method (B-F) and Treecode-Stretch to evaluate the space-charge field of increasing numbers of particles $N$. \hyperref[fig:performance_b]{Figure~\ref*{fig:performance_b}} shows the speedup of Treecode-Stretch relative to the brute-force method. \hyperref[fig:performance_c]{Figure~\ref*{fig:performance_c}} shows the elapsed time of brute-force method normalized to $N^2$. \hyperref[fig:performance_d]{Figure~\ref*{fig:performance_d}} shows the elapsed times of treecode method normalized $N\log N$. \added[id=R2-2]{\hyperref[fig:performance_e]{Figure~\ref*{fig:performance_e}} shows the relative error of treecode.} The treecode is performed with the fixed admissibility parameter $\eta=0.5$ and different interpolation degrees $n=2,3,4,5,6,7$. The maximum number of particles in the leaf cluster is chosen as $N_0=(n+1)^3$.} 
\label{fig:performance}
\end{figure}

\section{Conclusion}
In this study, based on the Lagrangian interpolation, we formulate a treecode for computing the relativistic space-charge field. We propose two approaches to control the interpolation error: Treecode-Stretch and Treecode-AVGRF. Our numerical result shows that Treecode-Stretch has better accuracy than Treecode-AVGRF for particle beams with momentum spread. Also, the performance of our treecode scales like $O(N\log N)$ as expected by the theory. Comparing to FMM, which usually relies on a sophisticated data structure, the proposed treecode is easy to implement and intrinsically parallelizable; it may be especially suitable for building in-house solvers for the study of space-charge effects from relativistic beams. For future work, we plan to extend our treecode formulation to FMM and investigate its parallelization under many-cores architectures like graphics processing units (GPU). \added[id=R2-4]{Some former works,~\emph{e.g.}, Ref.~\cite{wilson2021gpu-accelerated}, might provide a possible direction.} 

\appendix 
\section{Approximation of the denominator of the relativistic kernel function}\label{appendx:approx_kernel_denominator}
\begin{lemma}
Let $\bm{p}:=p_r\bm{e}_r + p_{z}\bm{e}_{z}$ be a vector in cylindrical coordinates. If $p_{z}\gg p_r$, we have the approximation 
\begin{equation*}
    \Vert \bm{x}\Vert^2_2 + (\bm{p}\cdot\bm{x})^2\approx r^2 + (1+p^2_z)z^2, \quad
    \text{where}\quad r=x^2 + y^2.
\end{equation*}
\end{lemma}
\begin{proof}
In cylindrical coordinates, the function can be expressed as
\begin{equation}\label{eq:kernel_denominator}
     \Vert \bm{x}\Vert^2_2 + (\bm{p}\cdot\bm{x})^2 =
    r^2 + z^2 + p^2_{r}r^2 + p^2_{z}z^2 + 2p_{r}p_{z}rz.
\end{equation}

To find out the approximation, we analyze the term
\begin{equation}\label{eq:kernel_denominator_second}
    p^2_{r}r^2 + p^2_{z}z^2 + 2p_{r}p_{z}rz
\end{equation}
in three main different cases.

For the cases of $r \sim z$ and $r \ll z$, we first recast \eqref{eq:kernel_denominator_second} to
\begin{equation*}
    p^2_z z^2\left(\frac{p^2_r}{p^2_z}\frac{r^2}{z^2}+1+2\frac{p_r p_z}{p^2_z}\dfrac{r^2}{z^2}\right)
\end{equation*}
and we can conclude that \eqref{eq:kernel_denominator_second} can be approximated by $p^2_z z^2$.

For the case $r \gg z$, we recast \eqref{eq:kernel_denominator_second} to
\begin{equation*}
    p^2_r r^2\left(1+\frac{p^2_z}{p^2_r}\frac{z^2}{r^2}+2\frac{p_z}{p_r}\dfrac{z}{r}\right)
\end{equation*}
and consider three further scenarios: 
\begin{itemize}
    \item if $\frac{p_z}{p_r}\frac{z}{r}\gg 1$, \eqref{eq:kernel_denominator_second} is approximately equal to $p^2_z z^2$
    \item if $\frac{p_z}{p_r}\frac{z}{r}\ll 1$, \eqref{eq:kernel_denominator_second} is approximately equal to $p^2_r r^2 \approx p^2_r r^2 + p^2_z z^2$ 
    \item if $\frac{p_z}{p_r}\frac{z}{r}\sim 1$, \eqref{eq:kernel_denominator_second} is approximately equal to $p^2_r r^2(\frac{p^2_z}{p^2_r}\frac{z^2}{r^2}+3)=p^2_z z^2 + 3p^2_r r^2$
\end{itemize}

Finally, we can conclude that:
\begin{itemize}
    \item For $r\sim z$, $r\ll z$ or $r\gg z \land \tfrac{p_z z}{p_r r}\gg 1$, \eqref{eq:kernel_denominator} can be approximated by
        \begin{equation*}
            r^2+(1+p^2_z)z^2 
        \end{equation*}
    \item For $r\gg z \land \tfrac{p_z z}{p_r r}\ll 1$,         \eqref{eq:kernel_denominator} can be approximated by
        \begin{equation*}
            (1+p^2_r)r^2+(1+p^2_z)z^2\approx r^2+(1+p^2_z)z^2
        \end{equation*}
    \item For $r\gg z \land \tfrac{p_z z}{p_r r}\sim 1$
        \begin{equation*}
            (1+3p^2_r)r^2+(1+p^2_z)z^2\approx r^2+(1+p^2_z)z^2 \qedhere
        \end{equation*}
\end{itemize}
\end{proof}

\section{Some properties of special relativity}
This section summaries some consequences of special relativity, which are already discussed in the literature, for example in the textbook of classical electrodynamics~\cite{jackson1999classical}. 
Consider one inertial frame $\mathcal{K}'$ moving with a velocity $\bm{u}$ (corresponding to the momentum $\bm{p}_u$) relative to another frame $\mathcal{K}$. The space-time coordinates of an event in these two frames follow the transformation
\begin{equation}\label{eq:space_time_transformation}
    (c_{0}t') = \gamma_u (c_0t) - \bm{p}_u\cdot \bm{x},\quad
    x'_{\parallel} = \gamma_u x_{\parallel} -p_u(c_0t), \quad
    x'_{\perp} = x_{\perp}
\end{equation}
where $\parallel$ and $\perp$ denote the components parallel and perpendicular to $\bm{p}_u$. Here, we use $p_u$ to denote the magnitude of $\bm{p}_{u}$ (\emph{i.e,} $p_{u}:=\Vert\bm{p}_{u}\Vert_2$). Similar to the space-time coordinates, the four-momentum (energy and momentum) in $\mathcal{K}'$ and $\mathcal{K}$ follows the transformation 
\begin{equation}\label{eq:4momentum_transformation}
    \gamma' = \gamma_u\gamma - \bm{p}_u\cdot\bm{p}, \quad
    p'_{\parallel}=\gamma_u p_{\parallel} - p_u\gamma,\quad
    p'_{\perp} = p_{\perp}.
\end{equation}

The transformation formulas for the position and momentum above are expressed in the component-wise form because of its convenience for the theoretical derivation. These transformations can also be expressed in vector form
\begin{align}
    &\bm{x}'=\bm{x}+\dfrac{1}{\gamma_{u}+1}(\bm{x}\cdot \bm{p}_u)\bm{p}_{u} - c_{0}t\bm{p}_{u},\\
    &\bm{p}'=\bm{p} + \dfrac{1}{\gamma_{u}+1}(\bm{p}\cdot\bm{p}_{u})\bm{p}_{u} - \gamma\bm{p}_{u}.
\end{align}




\section*{Acknowledgement}
This work was supported by DASHH (Data Science in Hamburg -- HELM\-HOLTZ Graduate School for the Structure of Matter) with the Grant-No.\ HIDSS-0002 and in part by the European Research Council under the European Union's Seventh Framework Programme (FP7/2007-2013) through Synergy Grant (609920). The authors acknowledge the computational resources of the Maxwell Cluster operated at Deutsches Elektronen-Synchrotron (DESY).


\bibliographystyle{elsarticle-num}
\bibliography{references}







\end{document}